\documentclass[a4paper,11pt]{article}
\usepackage{jheppub} % for details on the use of the package, please see the JINST-author-manual
\usepackage{lineno}
\usepackage{physics}
\usepackage{amsthm}
\usepackage{float}
\usepackage{comment}
\usepackage{tikz}
\usepackage{tikz-3dplot}
\usetikzlibrary{patterns}
\usetikzlibrary{decorations.pathmorphing, calc}
\usepackage{subcaption}
\usepackage[normalem]{ulem}

\theoremstyle{definition}

\theoremstyle{theorem}
\newtheorem{theorem}{Theorem}[section]
\theoremstyle{lemma}
\newtheorem{lemma}{Lemma}[section]

\nolinenumbers
\definecolor{babyblue}{rgb}{0.54, 0.81, 0.94}
\definecolor{softorange}{rgb}{0.95, 0.65, 0.35}
\definecolor{babyblue-dark}{rgb}{0.35, 0.65, 0.85}

\newcommand{\ii}{\mathrm{i}}

\title{Approximate locality, black hole complementarity and overlapping qubits}

\author[a]{ChunJun Cao,}
\author[a]{Gong Cheng,}
\author[b]{Alexander Jahn,}
\author[a]{Thomas Koutsikos}

\affiliation[a]{Department of Physics and Virginia Tech Center for Quantum Information Science and Engineering, Virginia Tech, Blacksburg, Virginia 24061, USA}

\affiliation[b]{Department of Physics, Freie Universit\"at Berlin, 14195 Berlin, Germany}

\emailAdd{cjcao@vt.edu}
\emailAdd{gongc@vt.edu}
\emailAdd{a.jahn@fu-berlin.de}
\emailAdd{kthomas01@vt.edu}

\abstract{We construct a toy model of an evaporating black hole using approximately local degrees of freedom acting on ``overlapping" qubits in which a version of black hole complementarity arises naturally. The operators corresponding to the radiation and the interior are identified as two distinct representations of the same fundamental algebra, thereby preventing the exact factorization of the Hilbert space into interior and exterior and avoiding the conventional no-cloning violations. We show how this toy model captures several qualitative and quantitative features of black hole evaporation and how the ability to account for this ``overlap" in the entropy calculation leads to the recovery of a Page curve.} 

\begin{document}
\maketitle

\section{Introduction}
\label{sec:intro}

A na\"ive application of unitary quantum mechanics to an evaporating black hole and its interior allows the information in the interior to be recoverable from the exterior Hawking radiation when enough of the black hole has evaporated. As a consequence, both the Hawking radiation and the interior would appear to encode the same information, violating the no-cloning theorem \cite{Wootters1982,DIEKS1982271}. In response, Susskind, Thorlacius, and Uglum \cite{Susskind_1993} proposed the idea of black hole complementarity, where, drawing from analogy with the Heisenberg uncertainty principle, they argued that this does not lead to inconsistencies as long as no single observer could verify the existence of cloning. However, unlike the Heisenberg principle, which one can now understand from the fundamental laws of quantum mechanics, the underlying principle that gives rise to such a black hole complementarity is still far from understood. The subsequent formulation of the firewall paradox \cite{Almheiri_2013} further sharpened the tension arising from the assumptions behind black hole complementarity. It showed that conditions that appear physically natural, specifically unitarity, the validity of semi-classical effective field theory, and ``no drama'' at the horizon, are mutually incompatible. Despite recent advances, its resolution remains controversial. For example, proposals in Ref.~\cite{Harlow:2013tf,Kim:2020cds} suggest that the complexity in decoding and in verifying the presence of cloning may be a promising direction for identifying the aforementioned fundamental principles that give rise to black hole complementarity. 

An earlier alternative argument approached this paradox from the point of view of non-locality. For example,  Ref.~\cite{Lowe_1995} argues that microcausality breaks down in such scenarios where spacelike separated observables that are in the exterior and the interior of the black hole no longer commute.  Refs.~\cite{Giddings_2010,Donnelly_2016, Giddings_2013} later suggested that  a strong locality violation need not be necessary and a non-violent form of locality violation may be sufficient to solve the information problem. 
Nevertheless, a natural theoretical framework that allows for a weak but tunable non-locality is still lacking --- existing proposals must either rely on the language of a local quantum field theory where exact locality was built in by default, and contort it with weak non-locality put in by hand, or by imposing gauge constraints where the appearance of locality is often obscured. 

Here we import an explicit quantum information-theoretic formulation called \emph{overlapping qubits} by Chao, Reichardt, Sutherland, and Vidick (CRSV) \cite{OQ} that could be valuable for a deeper understanding of approximate locality. These constructions provide a physically intuitive middle ground, where the exact tensor factorization in the operator algebra can be relaxed in a well-controlled fashion while maintaining the appearance of local operators. 

Applying it to a toy model of an evaporating black hole, we show that it provides a simple explanation why the appearance of cloning can emerge even within a unitary process. The core tenet is to allow simple, i.e., low-weight, operators that are spacelike separated to only be almost commuting, thereby violating strict microcausality, consistent with proposals like \cite{Donnelly_2016,Giddings_2013, Donnelly_20162}. However, this violation may require a prohibitive amount of experimental resources to identify, which is spiritually similar, but not identical, to the standard complexity protection arguments \cite{Harlow:2013tf,Akers:2022qdl}. 
In doing so, disjoint spatial subregions only approximately correspond to tensor factors, or more generally, commuting subalgebras. The breakdown of exact tensor factorization then allows us to circumvent the violations of no-cloning and entanglement inequalities. 
Surprisingly, the small overlaps also permit the substantial compression of a larger number of overlapping qubits into a much smaller Hilbert space, similar to the holographic principle where the volumetrically scaling degrees of freedom can seem to be compressed into a number of quantum degrees of freedom that are only area-scaling.

We construct a simple kinematic toy model by first positing a finite-dimensional fundamental Hilbert space spanned by the black hole microstates. We assume that some UV-complete theory of quantum gravity governs the unitary dynamics over this space. We then select quasi-local degrees of freedom in the form of overlapping qubits (or fermions) from this Hilbert space to mimic the approximate local effective field theory degrees of freedom. These overlapping qubits are split into two complementary sets that respectively describe the interior and radiation degrees of freedom. Then, by a simple linear algebra argument, we show that the interior operators can be reconstructed from the radiation degrees of freedom at late times. In other words, both the interior and exterior operators are explicitly realized as different representations of the same fundamental algebra --- there is no cloning because both the apparent radiation and the apparent interior operator algebras are equivalent up to a basis transformation.

Furthermore, we show that a Hawking-like entropy growth can be recovered for a naive observer whereas the Page curve can be recovered if the reconstruction takes into account the operator overlaps. Due to our technical limitations, we first demonstrate these principles explicitly using overlapping Majorana fermions; then we argue that a similar quantitative outcome can apply to more general formulations also. 

In Sec.~\ref{sec:overoper} we review the CRSV construction of overlapping qubits and fermions with minor adaptations for our current work. Then we introduce a toy model of an evaporating black hole in Sec.~\ref{sec:appbh} and discuss its physical implications qualitatively. In Sec.~\ref{sec:ferm_pagecurve}, we introduce more refined measures of entropy and characterize their behaviors in a system of overlapping fermions. We then discuss the limitations of this toy construction as well as generalizations in Sec.~\ref{sec:generalOQ}. Finally, we conclude with a few remarks on future work.  

\section{Subsystems and overlapping operators}
\label{sec:overoper}
Before we discuss black holes, we need to re-examine our intuitions of subsystems in quantum mechanics. Traditionally, subsystems are automatically formulated as tensor factors of the Hilbert space with broader generalizations as commuting subalgebras. 
For example, one of our default assumptions is that the Hilbert space of an $N$-qubit system is the tensor product of $N$ two-dimensional Hilbert spaces. 
This is natural as it allows us to perform operations on one constituent without altering the state of the others. 

However, physical subsystems rarely decompose in such an ideal way.  Realistic qubits can be subjected to weak long-range interactions, leakage, and imperfect control, all of which introduce small but systematic violations of subsystem independence, effectively making factorizations and commutativity approximate.  For example, focused laser beams used for single-qubit gates on trapped ions can easily leak onto neighboring ions, causing unintended rotations \cite{Parrado_Rodr_guez_2021,HAFFNER_2008,Piltz_2014}. Similar unintended interdependencies are expected in superconducting qubits due to crosstalk \cite{Sarovar_2020,crosstalksup}. In neutral-atom arrays, crosstalk and the Rydberg blockade --- depending on the way a qubit is encoded --- can also lead to unintended influence on distinct qubits \cite{Xia_2015,Henriet_2020,Saffman_2016}. As a result, nominally local operators acting on one qubit can have small but nonzero effects on others. 

More precisely,  operators acting on distinct isolated qubits are assumed to commute: if  operators $O_i, O’_j$ act on different qubits $i$ and $j$, then $[O_i,O’_j]=0$. 
However, for realistic systems, it is reasonable to ascribe to them operators where those acting on one qubit can have extremely small but nonzero noncommutativity with another.  If we describe these realistic single-qubit operators as $\{\tilde O_i\}$, $\{\tilde O'_j\}$ acting on physical qubits $i, j$, then the weak interdependencies discussed above can lead to small non-commutativity $\Vert[\tilde O_i,\tilde O'_j]\Vert<\epsilon$. This shows that our mathematical description of real qubits $i,j$ in terms of an exact tensor product of Hilbert spaces and unitary operators is only approximate.
When $\epsilon$ is sufficiently small, it can give one the illusion of having $N$ independent qubits when in reality the size of the Hilbert space can be exponentially smaller than $2^N$. 

Interestingly, similar approximate factorizations\footnote{Note that because overlapping operators generally do not form commuting subalgebras, this approximate factorization is not the same as simply generalizing to a subalgebra description where spacelike separate regions do not correspond to tensor products but to commuting subalgebras. } or microcausality violations can also appear in quantum gravity, raising the question whether exact factorizations are even possible at the fundamental level. In gravitational systems, local algebras need not strictly commute at spacelike separation due to gauge constraints and nonperturbative effects \cite{Donnelly_2016,Penington:2019kki,Akers:2022qdl,Balasubramanian:2022gmo,Balasubramanian:2022lnw,Antonini:2024yif,Almheiri_2013,Papadodimas_2014,Kabat_2014}. 
From a heuristic discretized perspective, placing a quantum field on a lattice and assigning qubits to lattice sites, quantum gravity generically induces overlaps that spoil exact tensor products. Similar ideas have also been used to understand the degrees of freedom counting problem in black holes and holography~\cite{Akers:2022qdl,Friedrich_2024,Cao_2025}. 
\subsection{Defining overlapping qubits}
\label{sec:overqubdisc}

The CRSV overlapping qubit \cite{OQ} construction realizes one concrete example of such ``approximate tensor factorization''. There can be other constructions of overlapping qubits beyond CRSV \cite{Friedrich_2024}, but we shall focus on their example for the sake of simplicity.

Let $\{P^{(I)}_i\}:=\{X^{(I)},Y^{(I)},Z^{(I)}\}$ be Pauli operators labelled by some index $I\in [N]$ and $j\in\{1,2,3\}$ correspond to the Pauli $X, Y$, and $Z$ operators respectively. Overlapping qubits are defined as operator sets that satisfy

\begin{align}[P^{(I)}_k,P^{(I)}_l]&=2 \ii \varepsilon_{klm}P^{(I)}_m\nonumber\\
\Vert [P^{(I)}_k,P^{(J)}_l] \Vert &={\cal O}(\epsilon)\label{eq:overlap},
\end{align}

These qubits are said to overlap because they almost but not exactly commute. However, a small overlap can lead to a drastic compression of Hilbert space dimension where CRSV demonstrated a construction where one can fit $O(e^{n\epsilon^2})$ such overlapping qubits into the Hilbert space of $n$ non-overlapping qubits, i.e., a $2^n$-dimensional Hilbert space. 

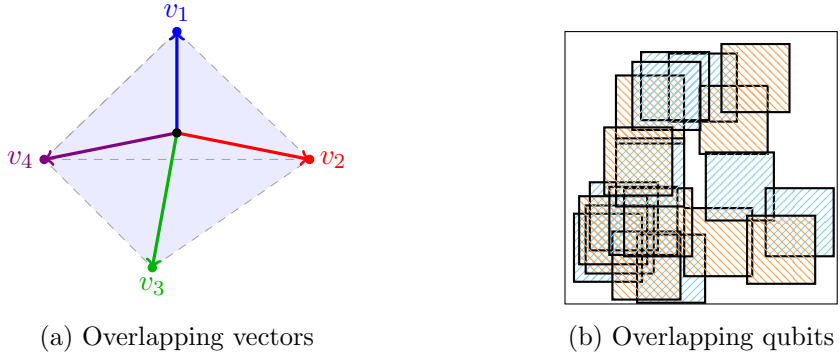
\begin{figure}[H]
\centering
% -------- Subfigure 1 --------
\begin{subfigure}{0.45\textwidth}
\centering
\begin{tikzpicture}[scale=1.3,line cap=round,line join=round]
    %------------------------------------------------
    % Vertices of tetrahedron
    %------------------------------------------------
    \coordinate (T) at (0,1.25);      % top
    \coordinate (R) at (1.35,-0.05);  % right
    \coordinate (L) at (-1.35,-0.05); % left
    \coordinate (B) at (-0.25,-1.15); % bottom
    % Center
    \coordinate (O) at (0,0.22);
    %------------------------------------------------
    % Soft face shading
    %------------------------------------------------
    \fill[blue!8]
        (T)--(R)--(B)--(L)--cycle;
    %------------------------------------------------
    % Dashed tetrahedron edges
    %------------------------------------------------
    \draw[dashed, gray!70]
        (T)--(R)--(B)--(L)--cycle;
    \draw[dashed, gray!70]
        (T)--(L);
    \draw[dashed, gray!70]
        (R)--(L);
    %------------------------------------------------
    % Vectors
    %-----------------------------------------------
    \draw[->, very thick, blue]
        (O)--(T);
    \draw[->, very thick, red]
        (O)--(R);
    \draw[->, very thick, green!70!black]
        (O)--(B);
    \draw[->, very thick, violet]
        (O)--(L);
    %------------------------------------------------
    % Central point
    %------------------------------------------------
    \filldraw[black] (O) circle (1.2pt);
    %------------------------------------------------
    % Vertex dots
    %------------------------------------------------
    \filldraw[blue]              (T) circle (1.2pt);
    \filldraw[red]               (R) circle (1.2pt);
    \filldraw[green!70!black]    (B) circle (1.2pt);
    \filldraw[violet]            (L) circle (1.2pt);
    %------------------------------------------------
    % Labels
    %------------------------------------------------
    \node[blue, above] at (T) {$v_1$};
    \node[red, right] at (R) {$v_2$};
    \node[green!70!black, below] at (B) {$v_3$};
    \node[violet, left] at (L) {$v_4$};
\end{tikzpicture}
\caption{Overlapping vectors}
\end{subfigure}
% % -------- Subfigure 2 --------
\begin{subfigure}{0.45\textwidth}
\centering
\begin{tikzpicture}[scale=0.3]
% -------------------------
% Parameters
% -------------------------
\def\L{12}        % side length of the large square
\def\a{3}         % side length of each small square
\def\N{20}        % number of small squares
% -------------------------
% Large square
% -------------------------
\draw (0,0) rectangle (\L,\L);
% -------------------------
% Randomly placed overlapping squares
% -------------------------
\foreach \i in {1,...,\N} {
    % Random position, chosen so the small square remains inside
    \pgfmathsetmacro{\x}{rnd*(\L-\a)}
    \pgfmathsetmacro{\y}{rnd*(\L-\a)}
    % Randomly choose one of the two diagonal patterns
    \pgfmathtruncatemacro{\pattern}{random(0,1)}
    \ifnum\pattern=0
        \filldraw[
            draw=black,
            thick,
            pattern=north east lines,
            pattern color=babyblue
        ]
        (\x,\y) rectangle ++(\a,\a);
    \else
        \filldraw[
            draw=black,
            thick,
            pattern=north west lines,
            pattern color=softorange
        ]
        (\x,\y) rectangle ++(\a,\a);
    \fi
}
\end{tikzpicture}
\caption{Overlapping qubits}
\end{subfigure}
\caption{An illustrative analogy. (a) The vector space ${\mathbb R}^n$ can fit up to $n$ mutually orthogonal vectors, but if we allow them to overlap we can fit more than $n$. (b) A Hilbert space with $\dim{\cal H}=2^n$ can fit up to $n$ independent qubits, but if we allow them to overlap we can fit more than $n$.}
\end{figure}

The intuitive reason for this lossy exponential compression comes from the Johnson-Lindenstrauss lemma~\cite{JohnsonLindenstrauss}, where one can identify $N=O(e^{n\epsilon^2})$ approximately orthogonal vectors with overlap of $O(\epsilon)$ in an $n$-dimensional vector space. 
By mapping these approximately orthogonal vectors onto Majorana fermions, one for each generator of the Clifford algebra, then this produces a set of $N$ ``overlapping fermions'' that mutually anti-commute up to $O(\epsilon)$ corrections. Finally, employing a fermion-to-qubit map, one can group and define overlapping Pauli operators that satisfy the algebraic and commutation relations in Eq.~\ref{eq:overlap}.

\subsection{Overlapping fermions}
\label{sec:overferm}
We now show how overlapping Majorana fermions can be constructed.
This is a slight modification of the construction in~\cite{OQ} that serves a dual purpose.
Not only does it enrich the variety of operators we can use in overlapping frameworks, but they are also more convenient for calculations.\\

Let $Cl(2n)=\langle \{C_i\}\rangle$  denote a $2^n$-dimensional (Hermitian) representation of the Euclidean Clifford algebra, 
\begin{equation}
    \{C_i,C_j\}=2\delta_{ij}\mathbb{I}_{2^n\times 2^n}.
\end{equation}
The generators $\{C_i\}$ are what one would use to describe a system of $2n$ Majorana fermions. 
The goal is to create $2N$ overlapping fermionic operators $\{\Psi_I\}$ that are Hermitian, approximately anticommute and square to the identity using the $2n$ fundamental fermions $\{C_i\}$. {Unless specified, throughout this paper capital indices will be ranging from $1$ to $2N$ and lowercase indices will be ranging from $1$ to $2n$, where $n<N$.}\\

We begin by invoking the Johnson-Lindenstrauss lemma,  which guarantees that we can find $2N$ approximately orthogonal unit vectors in ${\mathbb R}^{2n}$, where $N={\cal O}(e^{n\epsilon^2})$.
We gather these vectors as columns in the $2n\times 2N$ matrix $V$, or in other words:
\begin{align}
    V_{iI}&:=(\vec{v}^{(I)})_i\\
    V_{iI}V_{iJ}&=\vec{v}^{(I)}\cdot \vec{v}^{(J)}:=\delta_{IJ}+\epsilon R_{IJ}:=G_{IJ},\label{eq:vecs}
\end{align}
where $R_{II}=0$ and $R_{IJ}\sim O(1)$. In practice, the easiest way to construct such a matrix $V$ (which is not an isometry) is to start from an orthogonal matrix, randomly project out some rows and renormalize the columns so that they have unit norm.
In principle, one can add non-trivial structure by a careful construction of the $G$ matrix, however here we will focus on the structureless case. 
Therefore, we can now construct
\begin{align}
    \Psi_I=V_{iI}C_i.
\end{align}
One can easily verify given~\eqref{eq:vecs} that we are justified in calling the operators $\{\Psi_I\}$ overlapping fermions, since they are Hermitian, approximately anti-commute and square to the identity.
Thus, we have
\begin{equation}
    \{\Psi_I,\Psi_J\}=2\delta_{IJ}{\mathbb I}_{2^n\times 2^n}+2\epsilon R_{IJ}{\mathbb I}_{2^n\times2^n}.\label{eq:oferms}
\end{equation}

It is important to note that in this construction the overlapping operators are not acting on the same Hilbert space as their $2^N$-dimensional non-overlapping counterparts $C_I$. 
A useful modification that connects well with the non-overlapping description is to instead use a sub-algebra of a $2^N$-dimensional representation of the Euclidean Clifford algebra. 
More formally, let $\{C_I\}$ denote a $2^N$-dimensional representation of the Euclidean Clifford algebra:
\begin{equation}
    \{C_I,C_J\}=2\delta_{IJ}{\mathbb I}_{2^N\times 2^N}\label{eq:fullnonoverlapping}
\end{equation}
Let $A_{iJ}$ be a real $2n\times 2N$ matrix satisfying
\begin{equation}
    A_{iK}A_{jK}=\delta_{ij} .
\end{equation}
Without loss of generality (simply a matter of representation) we will be considering:
\begin{equation}
    A_{iJ}:=\delta_{iJ}=\begin{pmatrix}
        {\mathbb I}_{2n\times 2n}& \mathbf{0}_{ 2n\times 2(N-n)}
    \end{pmatrix}
\end{equation}
Then, the operators
\begin{equation}
    F_i:=A_{iJ}C_J=C_i\otimes{\mathbb I}_{2^{N-n}\times 2^{N-n}}\label{2.10}
\end{equation}
form a subalgebra and consequently, the operators
\begin{equation}
    {\tilde F}_I:=V_{jI} F_j
\end{equation}
satisfy
\begin{equation}
    \{{\tilde F}_I,{\tilde F}_J\}=2\delta_{IJ}{\mathbb I}_{2^N\times 2^N}+2\epsilon R_{IJ}{\mathbb I}_{2^N\times 2^N}. \label{eq:fulloverlapping}
\end{equation}
This defines overlapping fermions acting on the larger Hilbert space on which the $C_I$ act. The dimension of the subspace spanned by ${\tilde F}_I$ is $2^n$, matching with the dimension of the Hilbert space on which the $\Psi_I$ act.

We emphasize that this is by no means the most general toy model one can have.
There are many ways to construct overlapping fermions with vastly different overlap structures. 
While this particular choice using the Johnson-Lindenstrauss lemma and random projections is the most convenient for calculations in Sec.~\ref{sec:Page curve}, physically relevant models that can incorporate spatial locality, causal structure, and reproduce low-point expectation values will require more refined embedding maps for the overlapping vectors. 

\subsection{Overlapping qubits}
\label{sec:overlappingqubitsconstruction}
We now review how overlapping qubits are built from overlapping fermions in CRSV. The goal is to identify a set of Pauli operators $X,Y,Z$ for each qubit using fermionic operators. 
These operators should have all the desired properties locally --- each operator squares to the identity and anticommutes with other Pauli operators on the same qubit (c.f. equation \ref{eq:overlap}). However, operators belonging to distinct sets can have a small non-vanishing commutator.

To account for the extra structure, we prepare $6N$ approximately orthogonal vectors living in ${\mathbb R}^{2n}$ to define $N$ overlapping qubits. Once again, the existence of these vectors is guaranteed by the Johnson-Lindenstrauss lemma and for any two vectors $\ket{v_I}$ and $\ket{v_J}$ among them, we have
\begin{equation}
    \braket{v_I}{v_J}=\delta_{IJ}+\epsilon R_{IJ}\quad,\quad \epsilon \ll 1\quad,\quad R_{II}=0,
\end{equation}
and in this equation (and only this equation) the values for the capital latin indices range in $[1,6N]$.
Now, the vectors are separated in triplets and within each triplet we construct orthonormal bases $\{\ket{e}_I,\ket{f}_I,\ket{g}_I\}$. Therefore, for $s,t\in \{e,f,g\}$ we have:
\begin{equation}
    \braket{s_I}{t_J}=\delta_{IJ}\delta_{st}+\epsilon W^{st}_{IJ}\quad,\quad \epsilon \ll 1 \quad,\quad W^{st}_{II}=0\label{2.14}
\end{equation}
The notation we have chosen here is admittedly lax, but emphasizes how keeping track of the Gram-Schmidt procedure in terms of the original matrix $V_{iJ}$ is tedious. The important observation in the above formula is that vectors chosen from different triplets (which as we will see correspond to different qubits) will have an inner product of ${\cal O}(\epsilon)$.\\

We then construct three families of overlapping fermions as we showed earlier using a $2^n$ dimensional representation of the Euclidean Clifford algebra $C_i$:
\begin{equation}
    E_I=e_{jI}C_j\quad,\quad F_I=f_{jI}C_j\quad,\quad G_I=g_{jI}C_j,\label{2.15}
\end{equation}
where analogously to the previous construction we have organized all of the $\ket{e_I}$ vectors as columns of the $e_{jI}$ matrix and similarly for $f$ and $g$. One can easily verify that these matrices satisfy (again for $\{S,T\}\in\{E,F,G\}$):
\begin{equation}
    \{S_I,T_J\}=2(\delta_{IJ}\delta_{ST}+\epsilon W_{IJ}^{st}){\mathbb I}_{2^n\times 2^n}
\end{equation}
To finish the construction of the desired operators, we have the following Pauli-like operators (for clarity we will also show $Y$), where we need to suppress the summation convention.
\begin{align}
    X_I&=iE_IF_I\\
    Y_I&=iF_IG_I\\
    Z_I&=iE_IG_I
\end{align}
Indeed, these operators are Hermitian and for a given $I$ square to the identity, anticommute, and satisfy:
\begin{equation}
    [P_k^{(I)},P_l^{(I)}]=2i\varepsilon_{klm}P_m^{(I)}
\end{equation}
in the definitions of section~\ref{sec:overqubdisc}.

Following the overlap of the vectors, the commutator for Pauli-like operators acting on different overlapping qubits is non-zero, which indicates the fact that the qubits are not independent. Naturally, there could be a multitude of different constructions of overlapping qubits, either through a different fermion-to-qubit mapping or entirely independent from overlapping fermions. We use the one above because this is,  to the best of our knowledge, the only explicit overlapping qubit construction available. It is not unreasonable to expect that a better way to construct these overlapping qubits exists, but for the time being this is the construction that we will be using in the subsequent sections.
\section{Application to black hole evaporation}
\label{sec:appbh}

Having motivated the study of overlapping qubits in quantum gravity, we would like to see how they could be used to describe black hole evaporation. What we currently have is a simplistic toy model of overlapping qubits that manages to capture some conceptual features of black hole evaporation. Furthermore, within this model a notion of complementarity arises naturally and in a way that avoids the violation of the monogamy of entanglement. Our hope is that once they are better understood, overlapping qubits can be refined into a framework for studying black hole evaporation and maybe even quantum gravity. In this section, we will be introducing this toy model identifying which features of black hole evaporation it reproduces on a conceptual level and which require a better understanding.\\

\subsection{Roadmap from overlapping qubits to evaporating black holes}
\label{sec:roadmap}
We will begin this section by defining three ways of identifying degrees of freedom ---fundamental, overlapping/quasi-local, and effective --- relevant for our discussion.

As a starting point, let the Hilbert space of a black hole in quantum gravity be ${\cal H}_F$, with $\dim{\cal H}_F=2^n = e^{S_{\rm BH}}$.
Here $S_{\rm BH}$ is the entropy of the black hole. One can define an algebra of observables on these $n$ qubits that corresponds to the \textit{fundamental} degrees of freedom of quantum gravity. We further posit that the formation and evaporation of the black hole can be fully captured by a unitary process over ${\cal H}_F$ in quantum gravity.
At each time $t$, the fundamental description assigns a single, definite state $|\psi(t)\rangle_F$
 to the black hole.
This is consistent with the expectation from \cite{Donnelly_2016,Donnelly_20162,Marolf_2009,Jacobson_2019} where a fully unitary theory could be construed as living in the asymptopia of the spacetime containing an evaporating black hole while restricting to the space of states with  gauge invariant observables consistent with a black hole of entropy $S_{\rm BH}$.

In addition to these fundamental degrees of freedom we will also introduce $N(t)$ overlapping qubits which mimic the quasi-local degrees of freedom associated with the black hole interior and radiation. 
The time dependence comes from the changes in the interior and radiation degrees of freedom as the black hole evaporates. We will refer to these generally as the quasi-local degrees of freedom, which live in the fundamental Hilbert space as almost factorizable Pauli algebras. These would be the degrees of freedom that a (bulk local) observer not unlike ourselves  actually has access to. In particular, for beings that are resource and precision-limited so that they are unable to verify the existence of overlaps efficiently, the overlapping qubits will appear as real as non-overlapping qubits.

Finally, we can lift the $N(t)$ sets of Pauli operators to operators over a much larger $2^N$ dimensional Hilbert space ${\cal H}_{\text{eff}}$ and drop the small overlaps entirely. We will refer to this as the \textit{effective} description, where the local degrees of freedom are $N$ non-overlapping qubits. One should think of these playing the role of degrees of freedom in a local effective field theory that lives in a factorizable Hilbert space with respect to spatial regions.

Note that we have used terms like quantum gravity and effective field theory as motivating points of reference. In the current work, however, we will not focus in anyway on spatial locality or causal structure that is present in either theory for the sake of simplicity. We envision such structure will be embedded in the way that overlapping qubits are identified in the fundamental Hilbert space, e.g. through more structured linear maps $V$ and/or a wider variety of fermion-to-qubit maps. We refer to the map that identifies overlapping degrees of freedom in $\mathcal{H}_F$ as the emergence map $\mathcal E(t)$, which depends on time\footnote{One might expect the construction of such an emergence map to come out of \emph{quantum mereology} programs such as \cite{Carroll:2020gme,Loizeau:2023vxd,Loizeau:2024rrb,OliverWavePacket} but applied to notions of approximate tensor factorizations like overlapping qubits.}. The complete unitary dynamics on the fundamental theory will likely be given by a UV-complete theory of QG. And we assume that the fundamental and effective dynamics will fix the evolution of $\mathcal{E}(t)$ in a more comprehensive theory built out of overlapping qubits. We will construct a simplified kinematic model based on expectations of evaporating black holes and study our system at different ``snapshots''. 
\begin{figure}[t]
    \centering
    \includegraphics[width=0.95\textwidth]{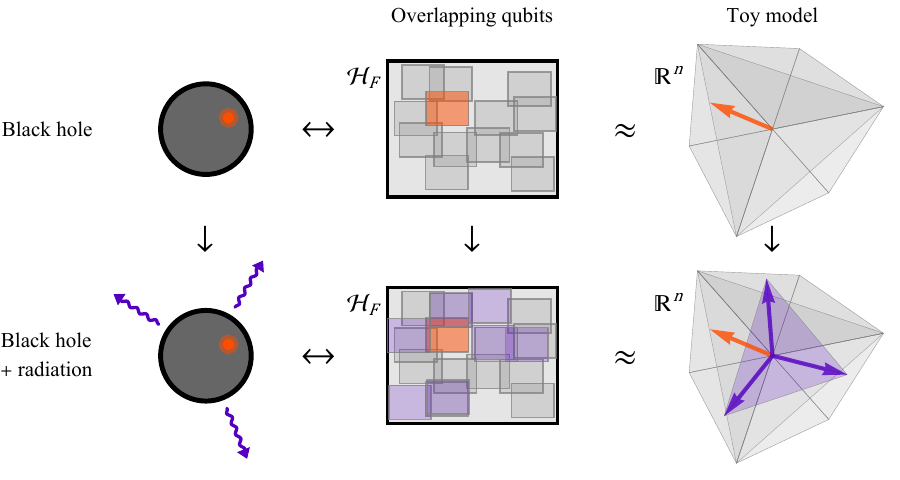}
    \caption{A visualization of the overlapping qubit (center) and random-vector toy model (right) introduced here.
    In our representation of black hole evaporation, we consider a fundamental Hilbert space ${\cal H}_F\neq {\cal H}_{BH}\otimes {\cal H}_R$ containing both local interior degrees of freedom (gray, with one highlighted in orange) and exterior radiation (blue), with both represented by overlapping qubits (overlapping boxes) and the latter forming an overcomplete basis around time $t \sim t_\text{Page}$.    
    Simplified in terms of random vectors in $\mathbb{R}^n$ that are approximately orthogonal, reconstruction of a given interior vector (orange) corresponds to building a span of exterior vectors (blue) that contains it. At large $n$, this requires a complicated basis representation, mimicking the difficulty of interior reconstruction.
    }
    \label{fig:OQBH}
\end{figure}

The remaining construction works as follows. We assign $N_I(t)$ overlapping qubits to the black hole interior and $N_R(t)$ to describe the exterior and radiation degrees of freedom. We assume that at $t=0$ the black hole just formed and has yet to begin to evaporate. Therefore, $\mathcal{E}(0)$ assigns $N_I(0)=N_0\geq n$ overlapping qubits to describe the black hole interior and $N_R(0)=0$ for the emitted quanta of radiation. For each unit of time $\delta t=1$, we identify two additional overlapping qubits in ${\cal H}_F$ --- one assigned to the interior (black) and the other to radiation (red) degrees of freedom (Fig~\ref{fig:OQBH}). They play the role of the operators associated with the ``Hawking pairs'' in the evaporation process.
By construction, the fundamental dynamics are assumed to be unitary and the individual overlapping degrees of freedom correspond to operators that act on different regions of space on a nice slice.

It is easy to see that since the overlapping qubits map to overcomplete bases spanning $L(\mathcal{H}_F)$, at ``Page time'' $t_p/\delta t\sim n$, enough of the black hole would have evaporated such that the interior algebra spanned by the $N_I(t_p)$ overlapping degrees of freedom can be fully reconstructed by the exterior operators. More precisely, for any interior operator $O_I$, there exist exterior operators $\{Q_J\}$ such that 
\begin{equation}
    O_I = \sum_J c_J^I Q_J
    \label{eqn:optrecon}
\end{equation}
and naturally, 

\begin{equation}
    O_I|\psi\rangle_F = \sum_J c_J^I Q_J|\psi\rangle_F.
\end{equation}
As an aside, notice that the second equation may be satisfied in a state-dependent manner before $t_p$.
In this sense, the exterior can access the same information as the interior as any exterior observer can manipulate the interior at will (and vice versa) by acting some interior operator on $|\psi_F\rangle$, so long as the dictionary (\ref{eqn:optrecon}) is known. In the mean time, for sufficiently large $S_{\rm BH}$, it is clear that the interior and exterior qubits have $\epsilon\ll 1$ overlap, hence giving these observers the illusion of locality as long as they only probe the commutator of low weight Pauli operators. While state-dependent reconstruction is of separate interest, we will focus on state-independent operator level reconstruction for this work.

Although the above operator reconstruction is a simple change of basis, it is practically difficult in a universe with a large black hole. In order for interior observers to influence the exterior (or vice versa) acausally, they generically have to construct an operator that is a coherent superposition of many operators in the interior or exterior which will likely require multi-qubit rotations that are of order entropy of the black hole. By contrast, one would expect causal signaling (from the exterior to interior) to be easy where one only has to be able to manipulate a much smaller number of overlapping qubits.

This overlapping qubit model provides a concrete instantiation of black hole complementarity. The interior and exterior degrees of freedom both describe the same system which admits a more succinct description in the fundamental Hilbert space. It is clear that there is no cloning of information --- both the interior and exterior are distinct representations of the same fundamental algebra past Page time and they are related to each other via a non-trivial transformation. However, due to the approximate factorization of ${\cal H}_F\neq {\cal H}_{BH}\otimes {\cal H}_{R}$, both the interior and the exterior would appear to have access to distinct qubits that are ``local'' to their respective domains if the small overlap is overlooked. 
Careless observers that ignore these overlaps, e.g., they work in the effective description, would indeed be under the illusion that information has been cloned because the exterior observers could proceed to ``decode'' the Hawking radiation. Notably, if they actually attempt this decoding, which is presumably a computationally costly operation (but not in the same class as the ones typically seen in complexity-protected setups like~\cite{Harlow:2013tf}) for someone living in the same universe, then the procedure will alter the state of the interior.
Indeed, this is similar to some ${\cal A}={\cal R}_{\cal B}$ proposals~\cite{Bousso_2013,Harlow:2013tf,Papadodimas_2013,susskind2013blackholecomplementarityharlowhayden} in spirit though quite distinct in realization. In particular, in this toy model the identification of the interior and the radiation Hilbert spaces is explicit, with the two seemingly independent subsystems being secretly two equivalent representations in the same fundamental Hilbert space.
\subsection{Complexity of operator reconstruction}

We will now explore the progression of complexity in operator reconstruction in an even simpler toy model for the encoding of interior and exterior degrees of freedom of an evaporating black hole into a single fundamental Hilbert space.
In this toy model, the Hilbert space is replaced by a simple vector space $\mathbb{R}^n$ of $n$ ``fundamental'' degrees of freedom, in which both exterior and interior modes are represented by vectors (the right side of Fig.~\ref{fig:OQBH}). 
Both from an interior and exterior perspective, the apparent growth of degrees of freedom --- local modes on an expanding spatial slice and continuing emission of Hawking quanta, respectively --- is modeled by adding random vectors (or rays, as we disregard the norm) in $\mathbb{R}^n$ (practically, by choosing some basis vector $\vec{v}_0$ and taking $O\, \vec{v}_0$ with $O$ sampled over the Haar measure for $SO(n)$). 

Choosing the perspective of the exterior observer, we have a set $\{w_k\}$ of random vectors $w_k \in \mathbb{R}^n$ that is growing linearly with time, and a single vector $v \in \mathbb{R}^n$ (also chosen randomly) as our ``interior mode'' that we wish to reconstruct.
Using linear combinations of the $w_k$, we try to get as close as possible to $v$. How close we can get (the \emph{fidelity} $f$ of our reconstruction) depends on the subspace spanned by the $\{w_k\}$: We simply apply a Gram–Schmidt process to $\{w_k\}$ to form an orthogonal basis $\{\tilde{w}_n\}$, and then compute $f=\left(\sum_k |v \cdot w_k|^2\right)^{1/2}$. We then expect $f \approx 1$ after $t=n$, i.e., after $\{w_k\}$ has grown to $n$ random vectors.
As $t>n$, the approximate basis becomes overcomplete: The entire $\mathbb{R}^n$ can be approximately reconstructed from any sublist of $n$ random vectors, with $v$ written as a linear combination of $n$ of the $w_k$ with a non-negligible prefactor. At $t>n$, we can therefore choose an orthonormal basis in different ways. To get the ``optimal'' one with the fewest high-weight terms, we apply a \emph{weighted} Gram–Schmidt process instead: 
From the full list, we choose the vector $w_\text{opt}^{(1)} \in \{w_k\}$ that maximizes $|v \cdot w_\text{opt}^{(1)}|^2$. We then build a new list 
\begin{equation}
    {w_n^\prime} = \left\{ \frac{w - w \cdot w_\text{opt}^{(1)}}{\Vert w - w \cdot w_\text{opt}^{(1)} \Vert } \;|\; w \in \{w_n\} \wedge w \neq w_\text{opt}^{(1)} \right\} \, ,
\end{equation}
which is the standard Gram-Schmidt orthogonalization step. Repeating this procedure iteratively produces a list of orthonormal basis vectors $o_\text{opt}^{(k)}$, terminating after $n$ steps (when all the remaining vectors in the list have been reduced to the zero vector).
We can then use a basis expansion
\begin{equation}
    v \approx \sum_k a_k \, w_\text{opt}^{(k)} \ ,
\end{equation}
which becomes an equality when $\{w_k\}$ reaches $n$ elements and where the basis coefficients $a_k = v \cdot w_\text{opt}^{(k)}$ decrease with $k$. 
How many basis vectors do we need to get a good approximation of $v$? This is quantified by the \emph{Krylov entropy}
\begin{equation}
    S_K = - \sum_k a_k^2 \log a_k^2 \ .
\end{equation}
Fig.~\ref{fig:krylov-entropy} shows $S_K$ for a list $\{w_k\}$ with $t$ elements (successively added to the previous list): At $t < n$, the weighted Gram-Schmidt process fails to fully reconstruct $\mathbb{R}^n$, so $\sum_k |a_k|^2 <1$ and the Krylov entropy is suppressed. It peaks around $t=n$, but then decreases: The overcomplete and non-orthogonal basis $\{w_k\}$ allows for an increasingly compact orthonormal basis to be constructed close to $v$, with the probability of any $w_k$ being close to $v$ increasing with the total number of elements in the list. 
However, this decay in Krylov entropy is slow: The higher the dimensionality $n$ of the vector space, the more unlikely it is to find any vector whose angle $\alpha$ with $v$ deviates much from $\pi/2$, as the probability of such an angle will decay with $\propto \sin^{n-2} \alpha$.
As an alternative notion of complexity of reconstruction in our toy model, we can also simply count the number of basis vectors after the weighted Gram-Schmidt process with a coefficient $a_k^2 > \epsilon$ above a truncation value $\epsilon$.
In Fig.~\ref{fig:krylov-entropy}, we see that this counting yields a similar decay to the Krylov entropy.

Though these measures of complexity differ from conventional notions of quantum circuit complexity for operators acting on a Hilbert space, we expect the following intuition from our toy model to hold more generally:
As discussed above, although the modes associated with exterior Hawking radiation are described by overlapping qubits, the full algebraic access to the interior is only possible after the Page time: At that point, a non-overlapping set of basis operators (the quantum analogue of our Gram-Schmidt orthonormalized basis) of the dimension of the fundamental Hilbert space can be formed with very high probability. 
However, the circuit complexity of recovering a single semiclassical mode from the interior at that point is still vast, corresponding to a superposition of many basis vectors in our toy model; consistent with considerations of black holes as efficient scramblers \cite{Kim:2020cds}, this would prevent a computationally bounded exterior observer from reconstructing the interior.
Our toy model suggests that after the Page time, the finiteness of the fundamental Hilbert space slowly reduces this complexity: the increasingly overcomplete basis formed by exterior modes allows an exterior observer to reconstruct an interior mode using shallower circuits.

We now comment on a few more speculative observations in this model. Even though the interior and exterior are encoded in the same Hilbert space, both are effectively separated by a complexity barrier:
a computationally bounded exterior observer acting on semiclassical modes cannot access those available to an interior observer, and vice-versa. However, as time progresses long after the Page time, this complexity horizon slowly retreats, gradually revealing the interior.
One can also invert this toy model to describe the experience of the \emph{interior} observer, where a similar situation takes place: extremal spatial slices through the interior that are anchored to an external observer at stationary distance from the horizon grow linearly with the latter's time \cite{Susskind:2014moa,Christodoulou:2014yia}.
As these slices avoid the singularity itself, an interior observer would thus appear to be able to access a set of EFT operators that is highly overcomplete in the fundamental Hilbert space, possibly allowing for low-complexity operations to appear on the exterior; however, in a practical physical setting with causal structure, such an interior observer would not actually be capable of accessing a large set of these overlapping modes, precluding them from imprinting any messages into the Hawking radiation that any exterior observer could easily read.

\begin{figure}
    \centering
    \includegraphics[width=0.99\linewidth]{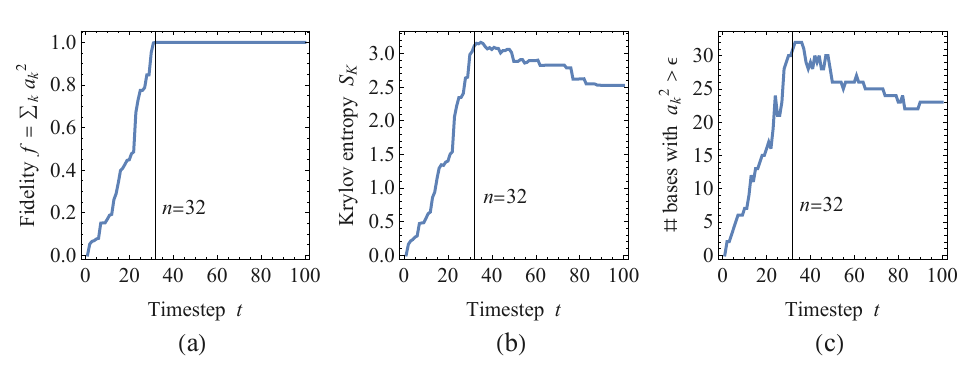}
    \caption{Representation of a random vector $v$ in a random-vector basis $\{w_k\}$ with $t$ elements under weighted Gram-Schmidt orthonormalization, for a vector space $\mathbb{R}^n$ of $n=32$ elements. 
    (a) Fidelity $f=\sum_{k=1}^t a_k^2$ of reconstruction of $v$ in the orthonormalized basis with basis elements $a_k$. 
    (b) Krylov entropy $S_K$ of  basis elements.
    (c) Number of basis vectors with $a_k^2 > \epsilon$ for a cutoff $\epsilon=10^{-3}$.
    }
    \label{fig:krylov-entropy}
\end{figure}

\section{Entropy and Page curve in overlapping fermions}
\label{sec:ferm_pagecurve}

Now that we have a conceptual framework for overlapping qubits, we proceed to check whether it can reproduce the Page curve, which is a key metric expected from black hole evaporation. Additionally we would like to understand the origin of Hawking's prediction in the semiclassical description~\cite{PhysRevD.13.191} and see operator reconstruction more explicitly. However, we immediately run into two problems.

The first one is that we are not technically equipped to perform direct analytical calculations using the original overlapping qubit mapping in~\cite{OQ}. As we saw in~\eqref{2.14} keeping track of the Gram-Schmidt procedure for all triplets is cumbersome and the expressions quickly become messy. Perhaps more importantly, equation~\eqref{2.15} breaks the summation convention explicitly. However, all these problems do not arise at the level of overlapping fermions as we saw in~\ref{sec:overferm}. As such, we will adapt the lessons we learned from the conceptual framework of overlapping qubits to a (Majorana) fermionic toy model more suitable for calculations in Section~\ref{sec:toymodels}. Although the approximate commutators are replaced by approximate anti-commutators, we believe that this gap, which can be overcome in future work, is a technical one and not a physical one. We will provide some justifications for this belief in the next section. In this section, however, one can avoid this subtlety by restricting the analysis to an even-parity fermionic state\footnote{By even parity we mean $P\rho P=\rho$ for $P$ being the fermionic parity operator.} and even-parity observables. 

The second problem is that it is not clear how one should define entropy in a way similar to the von Neumann entropies used in local effective field theories or factorizable systems when the notion of subsystem independence becomes relaxed. We will offer different definitions of entropy in sections~\ref{sec:funddesc},~\ref{sec:quasi-local} and discuss how they are expected to behave during the black hole evaporation process. We will then corroborate these expectations with numerics in section~\ref{sec:Page curve}.

\subsection{Fermionic toy model}
\label{sec:toymodels}
We will now work with a specific toy model of overlapping Majorana fermions. 
Let ${\cal F}=\{C_i\}$ be a $2^n$-dimensional representation of the Euclidean Clifford algebra. The operators in ${\cal F}$ will be identified with the $2n$ fundamental degrees of freedom of quantum gravity. We also saw that we need a quasi-local description of $2N(t)$ overlapping Majorana fermions ${\cal T}=\{\Psi_I\}$ to describe the interior and the radiation. To do that, we divide ${\cal T}$ in complementary sets, ${\cal I}=\{\Psi^{(\text{in})}\}$ for the interior and ${\cal R}=\{\Psi^{(\text{rad})}\}$ for the radiation, with $2N_I(t)$ and $2N_R(t)$ overlapping fermions in each set respectively. We are starting with all the overlapping fermions describing the interior at $t=0$ with $N(t)\geq n$. This partition changes at every time step $\delta t$ where {two new overlapping fermions are} added to each of ${\cal I}$ and ${\cal R}$ to model the quasi-local degrees of freedom.  Lastly, the effective description is obtained by dropping the $\epsilon$ terms in~\ref{eq:fulloverlapping}, which corresponds to lifting the overlapping fermions to a $2^{N(t)}$-dimensional representation of $2N(t)$ standard non-overlapping Majorana fermions acting on ${\cal H}_{\text{eff}}$, with $\dim{\cal H}_{\text{eff}}=2^{N(t)}$. For notational convenience, the time dependence will remain implicit from here on. Notice the doubling in $n\rightarrow 2n=n'$ and $N\rightarrow 2N=N'$ for the fermionic modes compared to qubits. For clarity, it can be helpful to think in terms of $n'$ and $N'$, which now play the role of the ``dimension'' for the fundamental and effective spaces.\\

Using the constructions in Sec.~\ref{sec:overferm},  it is now possible to see the operator reconstruction outlined in Sec.~\ref{sec:roadmap} explicitly. {We will do so explicitly for the generators of the Clifford algebra (i.e. the weight 1 operators) and similarly all higher weight operators will admit an expansion, since they are obtained by multiplication of the generators.}\\
We will separate the columns of the matrix containing the approximately orthogonal vectors and treat them as two matrices $\Omega$ and $W$:
\begin{align}
    \Psi_I^{in}=\Omega_{jI}C_j\quad,\quad \Omega_{kI}\Omega_{kJ}=\delta_{IJ}+\epsilon R_{IJ}\\ \Psi_I^{ext}=W_{jI}C_j\quad,\quad W_{kI}W_{kJ}=\delta_{IJ}+\epsilon Q_{IJ}
\end{align}
If the set of column vectors in $\Omega$ and $W$ form an overcomplete basis, then the rank of the $\Omega$ and $W$ matrices is equal to $n'$. In our construction, this is true for $\Omega$ at all times, but only true for $W$ after $t\sim n'$. For the anticommutator between interior and radiation we also have
\begin{equation}
    \{\Psi^{\text{(in)}}_{I},\Psi^{\text{(ext)}}_J\}=2\Omega_{kI}W_{kJ}{\mathbb I}=2\epsilon S_{IJ}{\mathbb I},
\end{equation}
indicating the small microcausality violation we mentioned in the introduction. Then, after Page time, when the overlapping fermions allocated to the radiation are sufficiently many to form an overcomplete basis of the fundamental algebra we can write any interior operator as an exterior operator and vice versa by:
\begin{align}
    \Psi_I^{\text{(in)}}=\Omega_{jI}(WW^T)^{-1}_{jk}W_{kJ}\Psi^{(\text{ext})}_J\nonumber\\
    \Psi_I^{\text{(ext)}}=W_{jI}(\Omega \Omega^T)^{-1}_{jk}\Omega_{kJ}\Psi^{\text{(in)}}_J
\end{align}
Once again, the lack of factorization means that information is not cloned in any way and every interior operator is eventually reconstructible in the exterior. Notably, this reconstruction is not possible for all operators before Page time, since the Gram matrices would not be full rank. This could be interpreted as information escaping the black hole in this fermionic toy model, but of course, we also want to study the entropy of the radiation and obtain the Page curve. Having introduced the fermionic toy model, we are ready to introduce three distinct definitions of entropy, corresponding to the fundamental, quasi-local and effective descriptions respectively. Each will offer a different perspective on the system.

The fundamental entropy can be thought of as an analog of the algebraic entropy in the UV theory of quantum gravity where a semi-classical description is lacking. For example, in AdS/CFT, one can think of the fundamental entropy as an algebraic entropy computed in the boundary theory where there need not be a (spatially) local description. The quasi-local and effective entropies are intended to simulate entropies that are derivable by an observer that work with a local-looking semiclassical description, e.g., a bulk description of quantum field theory on a background or a field theory coupled to semiclassical gravity. Of course, these analogies are not intended to be precise correspondences, but only conceptual guides that could help us better orient the motivations behind these different entropies.
\\

\subsection{Entropy in the fundamental description}
\label{sec:funddesc}
The cleanest entropy definition is in the fundamental description. By assumption, our system, which is said to be in the \emph{fundamental state} $\rho_F\in L(\mathcal{H}_F)$, is always well-defined in this description. One can also expand the state in terms of the Clifford algebra basis over the fundamental Hilbert space ${\cal H}_F$
\begin{align}
    \rho_F&=\frac{1}{d}\sum_{k=0}^{2n} i^{\frac{k(k-1)}{2}} x^{F}_{[i_1i_2\cdots i_k]}C_{[i_1}C_{i_2}\cdots C_{i_k]}\quad,\quad d=2^n\\
    \mathrm{where}\quad x^F_{i_1i_2\cdots i_k}&=i^{\frac{k(k-1)}{2}}\Tr(\rho_{F}C_{i_1}C_{i_2}\cdots C_{i_k}).
\end{align}
This state should be thought of as describing the black hole-radiation system in a UV theory of quantum gravity, e.g. a holographic CFT. In the above definition and the definitions to follow, the expansion coefficients are left arbitrary but they should satisfy constraints to ensure that the state is positive semi-definite (and pure if considering a pure state).

To define the entropy of the ``radiation'', it is natural to consider the algebraic entropy~\cite{Casini_2014,Harlow_2017}. Recall that given some von Neumann algebra $M$, for any state $\rho\in L(\mathcal{H}_F)$, one can always find a unique state $\rho_M \in M$ in the algebra such that
    \begin{equation}
      \Tr(\rho_M O_M)=\Tr(\rho O_M),\quad \forall O_M\in M.    
    \end{equation} 
    The algebra $M$ induces a canonical decomposition of the Hilbert space 
    \begin{equation}
        {\cal H}=\oplus_a {\cal H}_{{\cal A}_a}\otimes{\cal H}_{\bar{{\cal A}}_a}
    \end{equation}
    and the state $\rho_M$ admits the following canonical decomposition:
    \begin{equation}
        \rho_M=\oplus_a p_a\rho_{{\cal A}_a}\otimes \frac{{\mathbb I}_{\bar{{\cal A}}_a}}{|\bar{{\cal A}}_a|}
    \end{equation}
    where $p_a\geq 0$ are constants such that $\sum_a p_a=1$. 
    
    Then, the algebraic entropy of the state $\rho$ over the von Neumann algebra $M$ is the generalization of the von Neumann entropy on a subsystem when the Hilbert space does not factorize and is given by:
    \begin{equation}
        S(\rho,M)=-\sum_{a}p_a\log p_a+ \sum_a p_a S(\rho_{{\cal A}_a}).
    \end{equation}

The natural subalgebra $M$ in our case is the one generated by the overlapping fermions that correspond to the radiation degrees of freedom. Consider the set of overlapping operators $\mathcal{R}=\{\Psi^{(rad)}_I = V_{iI}C_i, I=1,\dots,t \}$ which are identified through a random construction:
each effective mode $\Psi^{(rad)}_I$ is chosen as a random Bogoliubov transformation of the fundamental modes $C_i$; this is mathematically similar to the random vector toy model introduced above, as the symmetry group of fermionic parity-preserving Bogoliubov transformations (i.e., transformations between pure fermionic Gaussian states) is given by $SO(2n)$ \cite{Windt:2020tra}.

At time $t\lesssim 2n$, the set $\mathcal{R}$ then contains $t$ overlapping fermions that are built out of near-orthogonal column vectors $V_{iI}$. %\TK{Isn't the last sentence incomplete?}
To identify an exact subalgebra that is closed, a Gram-Schmidt process is first used to produce an operator basis  $\{{\Psi}^\perp_{I}={(V^\perp})_{iI}C_i\}$ consisting of $\sim t$ generating elements where the column vectors of $({V}^\perp)_{iI}$ are orthonormal. Since they are unitary and pairwise anticommute, we can use them to generate an exact Clifford subalgebra $M$ of the full fundamental algebra \footnote{One may worry that ${\Psi}^\perp_I$ will have support on all $2n$ fundamental Clifford generators even at $t<2n$, however, it can always be related to an order $t$ subset of $\{C_i\}$ using an orthogonal transformation $O_{ij}$, which corresponds to a Gaussian unitary rotation over the space of states. Hence they form a proper set of Clifford generators. }. 

Suppose $\rho_F$ is pure, it is not hard to see that the entropy $S(\rho_F, M)=0$ at $t=0$ and for $t\gtrsim 2n$ because for sufficiently long time, the operator basis in $\mathcal{R}$ is complete. Then for intermediate times, the entropy follows a Page-like curve where it first rises, peaks when $t\sim n$, and then descents. 

For Haar random $\rho_F$, an analytical estimate follows from \cite{Bianchi_2019}, where if $\langle S_{page}(t) \rangle$ denotes the average Page entropy at $t$, then on average
\begin{equation}
    \langle S(\rho_{\rm haar},M)\rangle  = \begin{cases}
         \langle S_{page}(t)\rangle  \quad & \mathrm{if~t~even}\\
         \langle S_{page}(t-1)\rangle + \psi(2d_M+1)-\psi(d_M+1) \quad & \mathrm{if~t~odd}
    \end{cases}
\end{equation}
where $d_M$ is the dimension of the subalgebra $M$ and the difference between the two digamma functions are limiting to $\psi(2d_M+1)-\psi(d_M+1)\xrightarrow{d_M\rightarrow \infty} \log(2)$.
This analytical result is consistent with numerics, where a plot for $2n=16$ is shown in Fig.~\ref{fig:haar_page_curve}.
\begin{figure}
    \centering
    \includegraphics[width=0.7\linewidth]{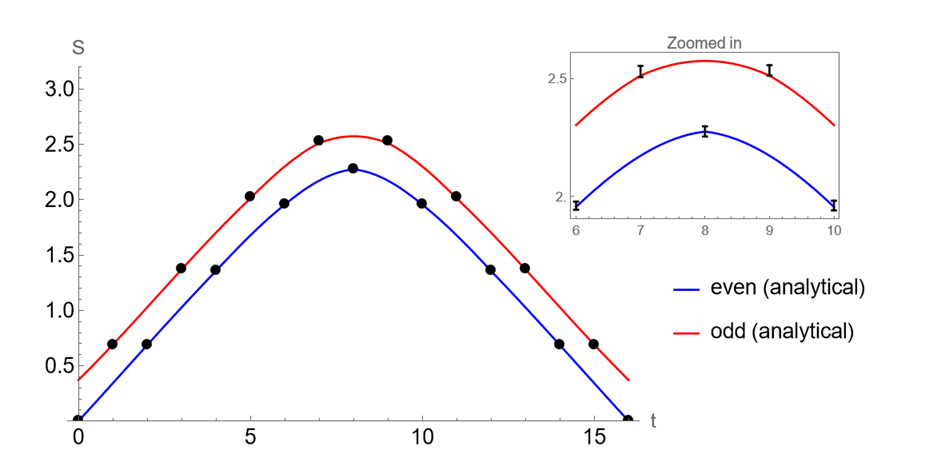}
    \caption{The algebraic entropy numerically averaged over $20$ Haar random states (black) is plotted the analytical results for $t$ odd (red) and even (blue) cases. Here $S$ denotes the algebraic entropy. A zoomed-in plot shows the error bars (standard deviations) that are too small to see in the full plot.}
    \label{fig:haar_page_curve}
\end{figure}

One might note a somewhat surprising feature that the curve returns to zero at  ``Page time'' --- when the exterior radiation can reconstruct the interior --- while the entropy peaks at half of the Page time. Although this mismatches our expectation from semi-classical physics, where the entropy peaks at Page time, the fundamental entropy is also supposed to correspond to a some UV theory of quantum gravity that has a rather non-local description of the system.

For the same reason, the recovery of a Page-like curve in the fundamental description is also not as informative because the truly non-trivial recovery of the Page curve should be from the semi-classical description where the description can be (quasi)-local.

\subsection{The (quasi-)local description}
\label{sec:quasi-local}
Things are more involved in the (quasi-)local description. Here, the observer is handed overlapping fermions 
without necessarily being told
that they overlap or that there is a mapping to a fundamental algebra of a smaller Hilbert space. At first sight, they may treat the fermions as independent and proceed from there.

Therefore, we motivate a different type of entropy construction from the following Gedanken experiment. Suppose the exterior observer is allowed to apply these overlapping operators in $\mathcal{R}$ to the state and can read out expectation values of any such operators any number of times. However, they are not given any information about the wavefunction of fundamental state, its algebra, or the mapping $V$. If they weren't informed that the fermions they can measure overlap, what would they conclude as the von Neumann entropy of the radiation degrees of freedom over time? 

One may think of this scenario as the case where an external observer is able to conduct the same experiment repeatedly and then ``measure'' these overlapping fermionic operators. Then they may reconstruct a state in $\mathcal{H}_{\rm eff}$ through quantum state tomography using these measurements while treating them as actual non-overlapping fermion measurements. Here by measurement, we simply mean that the observer can extract the operator expectation value to arbitrary degree of accuracy. We do not discuss, in this work, the deeper problem of measurement, such as resource-limited precision, observer degrees of freedom, or decoherence in the context of overlapping fermions. 
One may be tempted to define immediately the entropy of the ``radiation'' as the von Neumann entropy of this reconstructed state, but we would immediately run into problems as the state is not positive semidefinite in general. 

When confronted with this problem,  an observer may not question the overlapping nature of the fermions and insists that the larger effective Hilbert space is real. As a result, the observer may choose to ignore the negativity of the reconstructed state, identify the nearest physical state, and compute its entropy instead. We will call this the \emph{effective entropy} as the lack of overlap is most similar to the assumptions one would make in a local effective field theory. While it would certainly appear to be a strange and irresponsible choice if one is promised the immense power that can extract the expectation values exactly and can compute the entropy from it, this is not unreasonable in a more realistic experiment where noise and statistical fluctuations can also induce negativity  \cite{Smolin_2012}. 

A more careful observer on the other hand will conclude that something seriously wrong has happened to their assumption of the quantum mechanical model of the radiation. When equipped with a lot of computing and memory resource, we will see that they can learn that the state constructed is inconsistent with having a large Hilbert space with dimension $2^{|\mathcal R|}$ but instead a small Hilbert space of dimension $2^n$ when $|\mathcal R|>n$. Then, performing a suitable correction to the state, one can isolate the truly physical portion and compute its von Neumann entropy. We will refer to this as the \emph{quasi-local entropy}.

We claim that the effective entropy will be Hawking-like, where it grows as the black hole evaporates. On the other hand, the quasi-local entropy follows a Page curve.\\

\subsubsection{Pseudo-state reconstruction}
Let us now present this Gedanken experiment explicitly. Suppose the observers are measuring the overlapping fermions of~\eqref{eq:oferms} while thinking they have access to their non-overlapping counterparts. Then they would obtain the following pseudo-state, where $\tilde{\rho}$ is a Hermitian $2^{N_R}\times 2^{N_R}$ matrix with unit trace which is not necessarily positive semi-definite:
\begin{align}\label{state}
    {\tilde \rho}&=\frac{1}{D}\sum_{k=0}^{2{N_R}}i^{k(k-1)}\Tr(\rho_{\text{F}} \Psi_{[I_1}\Psi_{I_2}\cdots \Psi_{I_k]})C_{[I_1}C_{I_2}\cdots C_{I_k]}\nonumber\\
    &=\frac{1}{D}\sum_{k=0}^{2{N_R}}i^{\frac{k(k-1)}{2}}V^{\,\,\,[I_1}_{i_1}V^{\,\,\,I_2}_{i_2}\cdots V^{\,\,\,I_k]}_{i_k}x^F_{i_1i_2\ldots i_k}C_{I_1}C_{I_2}\cdots C_{I_k},
\end{align}
with 
\begin{align}
    x^F_{i_1i_2\cdots i_k}\equiv i^{\frac{k(k-1)}{2}}\Tr(\rho_{F}C_{i_1}C_{i_2}\cdots C_{i_k}).
\end{align}

Note that this time the antisymmetrization is more than a mathematical convenience to have an implicit summation notation. Since the fermionic operators no longer anticommute, a fermionic string (even with the corresponding $i$ prefactors) with a given ordering is generally not Hermitian. Furthermore, different orderings of the generators cease to be equivalent, since they can no longer be related to one another by anticommuting the Cliffords in the operators and expectation values accordingly. The antisymmetrization not only restores the Hermitianity of the operators in a way that is consistent with the non-overlapping description, but it can also be interpreted as averaging over all possible orderings with equal weight, which is a reasonable prescription if one wants to remain agnostic regarding the underlying ordering of the operators.

We can now notice that the pseudo-state ${\tilde \rho}$ has some characteristic structure by using the singular value decomposition (SVD) of $V$:
\begin{equation}
    V_{iI}=S_{ij}\Sigma_{jJ} O^T_{JI},
\end{equation}
where
\begin{equation}
    \Sigma_{iI}=\begin{pmatrix}
        \sqrt{(\Lambda_G)}_{2n\times 2n}& 0_{2n \times (2{N_R}-2n)}
    \end{pmatrix}\equiv (\sqrt{\Lambda}_G)_{ij}\delta_{jI}\quad,\quad G_{ij}:=V_{iK}V^T_{Kj}\label{eq:gmat}
\end{equation}
is the singular value matrix of $V$ and $\Lambda_G$ is the matrix containing the eigenvalues of $G$.
Let $U$ be the Gaussian unitary that generates the following orthogonal rotation among the Clifford generators in the large Hilbert space,
\begin{equation}
    UC_I U^\dagger = O_{IJ}C_J.
\end{equation}
 This allows us to study the pseudo-state ${\tilde \rho}_U\equiv U\tilde{\rho} U^\dagger$, which has the same spectrum as $\tilde{\rho}$ but can be expressed more compactly. 
Defining
\begin{equation}
    z_{i_1i_2\ldots i_k}\equiv S_{i_1j_1}S_{i_2j_2}\cdots S_{i_kj_k} x^F_{i_1i_2\ldots i_k},
\end{equation}
we see that we can express the transformed pseudo-state as:
\begin{equation}
    \tilde{\rho}_U=\frac{1}{D}\sum_{k=0}^{2{N_R}}i^{\frac{k(k-1)}{2}}(\sqrt{\Lambda}_G)_{j_1q_1}\delta_{q_1I_1}(\sqrt{\Lambda}_G)_{j_2q_2}\delta_{q_2I_2}\cdots (\sqrt{\Lambda}_G)_{j_kq_k}\delta_{q_kI_k} z_{i_1i_2\ldots i_k}C_{[I_1}C_{I_2}\cdots C_{I_k]}
\end{equation}
The $\delta_{jI}$ symbol is a slight abuse of notation to indicate that only the $2n$ non-trivial entries of each $\Sigma_{iJ}$ contribute in the sum. We can select a convenient representation for the generators of the Clifford algebra such that the indices of these $2n$ entries can be identified with the first $2n$ generators, i.e
\begin{equation}
    C_{I=i}=C_i\otimes {\mathbb I}_{2^{{N_R}-n}\times 2^{{N_R}-n}} \label{form}
\end{equation}
Thus, we have the final form of the transformed pseudo-state
\begin{align}
    {\tilde \rho}_U&=\left(\frac{1}{2^n}\sum_{k=0}^{2n}\sqrt{\lambda_1\lambda_2\cdots \lambda_{k}}z_{[i_1i_2\cdots i_k]}C_{[i_1}C_{i_2}\cdots C_{i_k]}\right)\otimes \frac{{\mathbb I}_{2^{{N_R}-n}\times 2^{{N_R}-n}}}{2^{{N_R}-n}}    \label{eq:recon}\\
    &:={\tilde \sigma}\otimes \frac{{\mathbb I}_{2^{{N_R}-n}\times 2^{{N_R}-n}}}{2^{{N_R}-n}},
\end{align}
where $\lambda_k$ are the eigenvalues of the matrix $G$ defined in~\eqref{eq:gmat}.

\subsubsection{Effective and quasi-local entropies from pseudo-state}
Equipped with the pseudo-state, we can now pretend to be the unsuspecting observer and compute its effective entropy. First we must find the closest valid density matrix, similarly to what happens in real life tomography setups given measurement-induced negativity.

This process can be recast as the following optimization problem. Let $L({\cal H}_{\text{eff}})$ be the space of bounded linear operators on ${\cal H}_{\text{eff}}$ and ${\cal D}({\cal H}_{\text{eff}})$ the space of density matrices on ${\cal H}_{\text{eff}}$, i.e:
\begin{equation}
    D({\cal H}_{\text{eff}})=\{\rho\in L({\cal H}_{\text{eff}})\,\,\text{s.t:}\,\,\rho^\dagger=\rho\,\,,\,\,\rho\geq 0\,\,,\,\,\Tr(\rho)=1 \}
\end{equation}
Then the negativity correction algorithm is the map:
\begin{align}\label{eq:maxlike}
\Phi &: L(\cal{H}_{\text{eff}}) \longrightarrow \cal{D}(\cal{H}_{\text{eff}}) \\
\Phi(\sigma)&:=\operatorname*{arg\,min}_{\rho \in {\cal D}({\cal H}_{\text{eff}})} \Vert\rho-\sigma\Vert_2^2.
\end{align}

Smolin et al \cite{Smolin_2012} solved this optimization problem explicitly by obtaining the maximum likelihood mixed state assuming Gaussian noise. We will be referring to this as the negativity correction algorithm and use it to identify the closest ``physical state''. 

We first set up two lemmas in regard to this algorithm, which we will prove in appendix~\ref{sec:proofs}.
\begin{lemma}
    Consider a trace 1 matrix $\sigma$, with real eigenvalues undergoing the negativity correction algorithm. 
    Let $\lambda_1$ be the largest eigenvalue and $\lambda_2$ the second largest eigenvalue. 
    Then the corrected density matrix $\rho=\Phi(\sigma)$ will be pure if $\lambda_1-\lambda_2>1$.\label{theorem1}
\end{lemma}
More importantly, the following lemma guarantees that the structure of our reconstructed matrix $\tilde{\rho}_U$ is preserved under the negativity correction algorithm.
\begin{lemma}
    Consider a unit-trace matrix $\tilde{\sigma}$ with real eigenvalues and let $\Phi(\tilde{\sigma})=\tilde{\rho}'$.
    Then the matrix 
    \begin{equation}
        \mu:={\tilde \sigma}\otimes \frac{1}{d}{\mathbb I}_{d\times d}
    \end{equation}
    satisfies
    \begin{equation}
        \mu':=\Phi(\mu)={\tilde \rho}'\otimes \frac{1}{d}{\mathbb I}_{d\times d}.
    \end{equation}\label{theorem2}
\end{lemma}
In other words, the action of performing the negativity correction algorithm on a state and the action of tensoring the state with a normalized identity block commute. This Lemma can be immediately applied to our transformed pseudo-state, mapping it to a valid density matrix of the form
\begin{equation}
    \mu'={\tilde \rho}'\otimes \frac{{\mathbb I}_{2^{{N_R}-n}\times 2^{{N_R}-n}}}{2^{{N_R}-n}}\quad,\quad {\tilde \rho}':=\Phi(\tilde{\sigma}).
\end{equation}

Then the effective entropy is simply the von Neumann entropy of $\mu'$
\begin{equation}
    S_{\text{eff}}:=S_{\text{vN}}({\mu'}).
\end{equation}
An unsuspecting observer who is also unaware of the compression map $V$
will not be able to easily resolve the identity block, but can only compute the entropy in a scrambled basis. As a result, they will find that 

\begin{equation}
    S_{\text{eff}} = S(\tilde{\rho}')+{N_R}-n \sim {N_R}-n \gg 1
\end{equation}
will increase linearly as the black hole evaporates.

In contrast, a more careful observer with no knowledge of the map $V$ but sufficient computational power would be able to identify that the truly physical degrees of freedom are encoded in ${\tilde \rho}'$ through analyzing the state's spectrum by exact diagonalization. Then they would learn that the actual dimension of the algebra that contributed to this physical process is much smaller than what it appears to be as the reconstructed state has no support over the identity block. Therefore, the entropy they would attribute to the state is the quasi-local entropy
\begin{equation}
    S_{\text{ql}}:=S_{\text{vN}}({\tilde \rho}'),
\end{equation}
which is clearly upper bounded by $n$. Note that since it is impossible to obtain the overlap matrix $V$ from measurements without also knowing the fundamental spectrum $x^F$, it is not possible to reconstruct the state $\rho_F$ exactly.

Indeed, as originally advertised, the presence of the identity block is merely an artifact of the observer's mistaken assumptions about the dimensionality of the Hilbert space, which can clearly be seen from the singular value decomposition. Its identification and removal  would lead an observer with sufficient computational power to use $S_{\text{ql}}$ rather than $S_{\text{eff}}$, which will be essential for obtaining the Page curve.

\subsection{Entropy curves}
\label{sec:Page curve}
To obtain the Page curve from $S_{\text{ql}}$ we still need to show that ${\tilde \rho}'$ is pure at the end of the evaporation and one avenue would be to prove that the Lemma~\ref{theorem1} applies. 
Unfortunately, finding the spectrum of the tomography matrix~\ref{eq:recon} has been technically challenging without adding additional structure. Therefore, we will restrict ourselves to fermionic Gaussian states for the rest of this section to simplify our calculations. 

The density matrix describing a fermionic Gaussian state can be brought to the following canonical form
\begin{equation}
    \rho=\bigotimes_{i=1}^n\begin{pmatrix}
        \frac{1+\nu_i}{2}&0\\
        0&\frac{1-\nu_i}{2}
    \end{pmatrix},
\end{equation}
where the $\nu_i$ are the singular values of the covariance matrix and are called \textit{Williamson eigenvalues} \cite{Bravyi:2004bud}.  For fermionic Gaussian  states they range in $\nu_i\in[-1,1]$. Pure states have singular values with $|\nu_i|=1$. These singular values can be found from observables through the covariance matrix, which is defined as
\begin{equation}
    \Gamma^{\rho}_{jk}=\frac{i}{2}\Tr\left(\rho[C_j,C_k] \right)=i\Tr(\rho C_jC_k)-i\delta_{jk},
\end{equation}
where $C_{j}, C_k$ are majorana operators.

This greatly simplifies the required measurements one would need to perform tomographically, as the spectrum is uniquely determined by the two-point correlators. Indeed, since the covariance matrix is real and antisymmetric, it admits the following block diagonal form through a rotation $R$ 
\begin{equation}
    \Gamma^\rho:=R\left( \bigoplus_{i=1}^n \begin{pmatrix}
        0&\nu_i\\
        -\nu_i&0
    \end{pmatrix}\right)R^T.
\end{equation}

To avoid keeping track of signs, we represent the matrix in a basis where all $\nu_i\geq 0$. 
The entropy of a Gaussian fermionic state can be expressed in terms of these singular values as:
\begin{equation}
    S(\rho_{2^n\times 2^n})=\sum_{i=1}^n h\left(\frac{1+\nu_i}{2} \right)\quad,\quad h(x)=-x\log x - (1-x)\log(1-x),
\end{equation}
therefore, the two point function contains all the information required for the entropy.\\

We assume that the state in the fundamental Hilbert space $\rho_F$ is a fermionic Gaussian state. Then the pseudo-state $\tilde\rho$ in the effective Hilbert space is also a Gaussian state, since all the $n$-point functions between $\tilde\rho$ and $\rho_F$ are related by linear transformation, as shown in Eq.~\eqref{state}, which preserves Wick factorization. 

To model an evaporating black hole, we identify $N_R(t)$ overlapping fermionic modes as the radiation degrees of freedom, and take $2N_R(t)=t$ to grow linearly with time. The pseudo-state $\tilde{\rho}(t)$ of these radiation modes is obtained from Eq.~\eqref{state}, with the corresponding state $\rho_F$ in the fundamental Hilbert space $\mathcal{H}_F$ being a random fermionic Gaussian pure state.

Since $\tilde{\rho}(t)$ is also Gaussian, it is fully determined by its two-point functions, or equivalently by its covariance matrix,
\begin{equation}\label{eq:Gammatilde}
\begin{split}
    \tilde{\Gamma}_{IJ}:=&\frac{i}{2}\Tr(\tilde\rho [C_I,C_J])=\sum_{ij}\frac{i}{2}\Tr(\rho_F[C_i,C_j])V(t)_{iI}V(t)_{jJ}\\
    =&\sum_{ij}V(t)_{iI}V(t)_{jJ}\Gamma_{ij}.
\end{split}
\end{equation}

The time-dependent overlap matrix $V(t)_{iI}$ is constructed as follows. We first generate a global $2n\times 2M$ matrix $V'$ by random orthogonal projection, with fixed $M>N_R(t)$ for all times considered. The column vectors in $V'$ satisfy the nearly orthonormal condition in Eq.~\eqref{eq:vecs}\footnote{Note that this mapping $V'$ only focuses on the embedding of the radiation degrees of freedom. However, it is straightforward to also construct a map that includes the interior degrees of freedom by attaching more columns to $V'$. Since this does not affect the radiation embedding,  we will only focus on the $N_R$ portion in this section.}. For each time $t$, we then obtain $V(t)$ by selecting $2N_R(t)$ columns of $V'$. Thus $V(t)_{iI}$, with $I=1,2,\cdots 2N_R(t)$ defines the overlap map for the radiation modes present at time $t$. 

From the construction, it is clear that when $N_R(t)>n$ the rank of  covariance matrix $\tilde\Gamma$ saturates at $n$. Therefore we expect at least $N_R-n$ zero Williamson eigenvalues. This implies that the corresponding pseudo-state takes the following form 
\begin{equation}\label{eq:rhotilde}
    S^T\tilde\rho S=\bigotimes_{I=1}^{n}\begin{pmatrix}
        \frac{1+\tilde\nu_I}{2}&0\\
        0&\frac{1-\tilde\nu_I}{2}
    \end{pmatrix}\bigotimes_{J=1}^{N_R-n} \begin{pmatrix}
        \frac{1}{2}&0\\
        0&\frac 1 2
    \end{pmatrix}
\end{equation}
for some $S\in SO(2N_R)$.  Here the $\tilde\nu_I$'s are the Williamson values of the pseudo-state.  This is consistent with the structure of identity block obtained in Eq.~\eqref{eq:recon}. 

Now we compute the quasi-local entropy and effective entropy from $\tilde\rho(t)$ numerically and analyze their dependence on time. 

As a first example, consider $2n=40$ fermions in the fundamental Hilbert space, and $M$ is chosen such that the average overlap $\overline{|V_I\cdot V_J|}=0.1$, for $I\neq J$. 
The pseudo state $\tilde\rho$ of the radiation modes is then adjusted using the negativity correction algorithm in Eq.~\eqref{eq:maxlike}.   The resulting entropy as a function of $t$ is shown in Fig~\ref{fig:page}. In the right panel, we also plot the number $k_g$ of good Williamson eigenvalues of the uncorrected reduced state, defined as those satisfying $\nu_i \in (0,1)$. 

\begin{figure}[t]
    \centering
    \begin{subfigure}{0.48\textwidth}
        \centering
        \includegraphics[width=\linewidth]{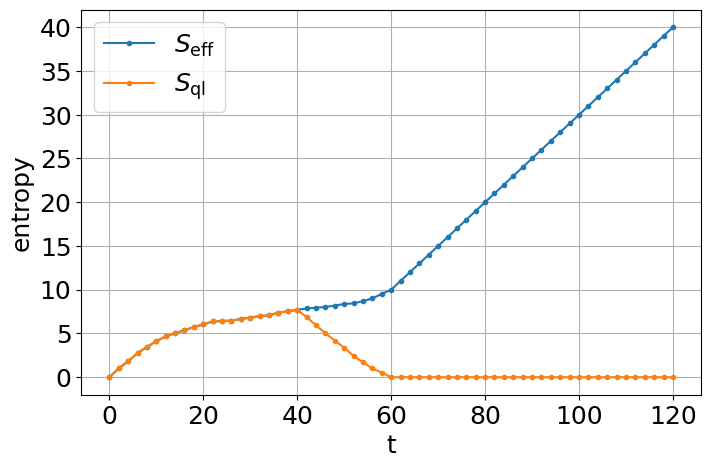}
        \caption{}
        \label{fig:Svst120}
    \end{subfigure}
    \hfill
    \begin{subfigure}{0.48\textwidth}
        \centering
        \includegraphics[width=\linewidth]{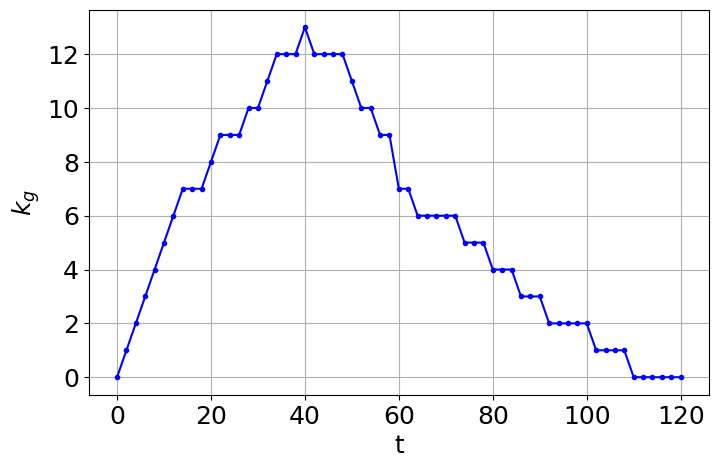}
        \caption{}
        \label{fig:good40}
    \end{subfigure}
    \caption{(a) Entropy of the radiation modes as a function of $t$. (b) The number of good Williamson values of the radiation state as a function of $t$. The average overlap $\epsilon=\overline{V_I\cdot V_J}=0.12$. Here $2n=40$ and a suitable $M$ is chosen to produce the specified average overlap.}
    \label{fig:page}
\end{figure}

For larger systems, an exact negativity correction becomes exponentially expensive. We therefore use an approximate procedure in which the spectrum of each subsystem is corrected independently, in analogy with a mean-field approximation. This gives the entropy curve shown in Fig.~\ref{fig:page2}, for a state $\rho_F$ of $2n=1000$ fermions in the fundamental space.  

Unlike the fundamental entropy, which returns to zero near the Page time $t \sim 2n$, the quasi-local entropy in these examples remains nonzero at the Page time and begins to decrease only shortly afterward. This property appears to be model-independent. Indeed, we observe a correlated trend between the quasi-local entropy $S_{\mathrm{ql}}$ and the number $k_g$ of good Williamson eigenvalues in both models above. 

We also comment on two model-dependent time scales $t_1<2n$ and $t_2>2n$ that we define to describe the behavior of the Williamson spectrum. For $t<t_1$, the number $k_g$ of good Williamson eigenvalues grows linearly with the number of fermionic modes in the radiation subsystem. In the intermediate regime $t_1<t<2n$, some Williamson eigenvalues begin to exceed $1$, but $k_g$ nevertheless continues to increase with the size of the radiation subsystem. At $t=2n$, the rank of the pseudo-state stops growing. Beyond this point, more Williamson eigenvalues move outside the physical range, and $k_g$ starts to decrease. Eventually, at $t=t_2$, all Williamson eigenvalues become larger than $1$. After the correction procedure, the radiation state is then pure, and this is precisely the point at which the quasi-local entropy vanishes. 

\begin{figure}[t]
    \centering

    \begin{subfigure}{0.48\textwidth}
        \centering
        \includegraphics[width=\linewidth]{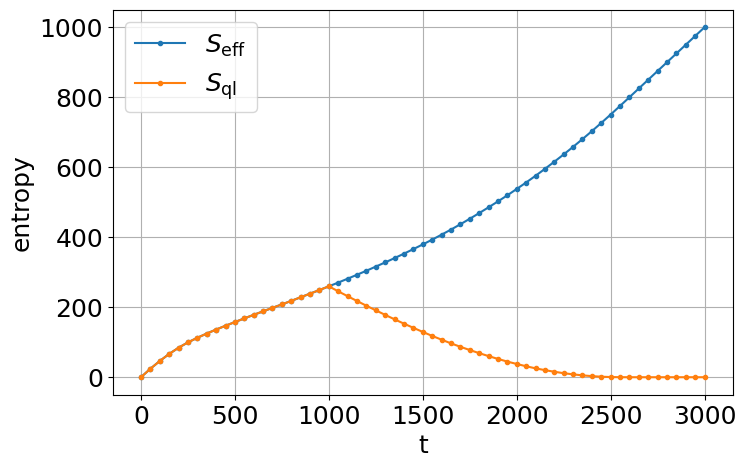}
        \caption{}
        \label{fig:Svst1500}
    \end{subfigure}
    \hfill
    \begin{subfigure}{0.48\textwidth}
        \centering
        \includegraphics[width=\linewidth]{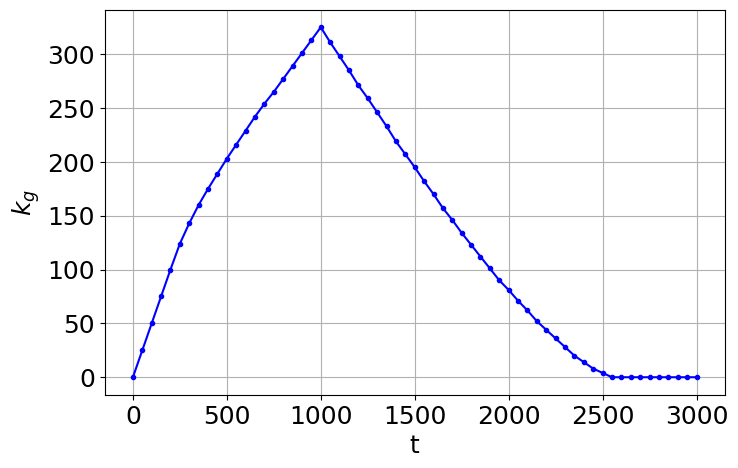}
        \caption{}
        \label{fig:good1500}
    \end{subfigure}

    \caption{(a) Entropy of the radiation modes as a function of $t$. (b) The number of good Williamson values of the radiation state as a function of $t$. The  number of  fermion in fundamental space is $2n=1000$. The average overlap $\epsilon=\overline{V_I\cdot V_J}=0.02$.}
    \label{fig:page2}
\end{figure}

Although the specific time $t_2$ depends on details of the model, we expect the quasi-local entropy to vanish at sufficiently late times as a universal feature.  In particular when $t=2M$, the radiation modes are described by the covariance matrix  $\tilde\Gamma$ in Eq.~\eqref{eq:Gammatilde} with all singular values of $V(t)$ greater than $1$ (Appendix~\ref{sec:singvalv}). In this regime, the  Williamson values of $\tilde\rho$ and the corrected state $\tilde{\rho}'$ have the following property.

\begin{lemma}\label{lm:william}
   Suppose the fundamental state $\rho_F$ is pure,  and the smallest singular value of $V(t)$ satisfies $\lambda_{\min}(V(t))\geq 1$. Then all the non-zero Williamson values $\tilde\nu_I$ of $\tilde{\Gamma}(t)$ satisfy $|\tilde{\nu}_I|\geq 1$. 
\end{lemma}

\begin{theorem}\label{lm:pure}
    Suppose that the assumptions of Lemma~\ref{lm:william} hold. Then, after applying the negativity-correction algorithm, the corrected pseudo-state $\tilde{\rho}'(t)$ is pure.
\end{theorem}

The proofs are given in Appendix~\ref{app:pure}.

\section{General overlapping operators}
\label{sec:generalOQ}
Although a simple construction of randomly overlapping fermions reproduces many desirable features of the entropy of an evaporating black hole, such a toy model inevitably overlooks many important physical details that are required for a more convincing QI description of the black hole information problem. 
\subsection{Limitation of the toy model }
A physical check is the approximate recovery of local effective field theory correlators. Although we did not impose any spatially local structures or more refined features of quantum field theories, one may check whether such a toy model can accommodate these constraints. 

In our fermionic toy model, we assumed that each single-particle effective operator can be represented as a linear combination of single-particle operators in the fundamental Hilbert space,
\begin{equation}\label{eq:single}
    \Psi_I =\sum_i V_I^iC_i.
\end{equation}
This relation allowed us to derive an embedding of a fundamental state into the effective Hilbert space; see Eq.~\eqref{state}. Conversely, given an effective state $\rho_{\mathrm{eff}}$, one may ask whether there exists a state $\rho_F$ in the fundamental Hilbert space that approximately reproduces the low-weight correlation functions of $\rho_{\mathrm{eff}}$:
\begin{equation}
    \Tr(\rho_F \Psi^{[I_1}\Psi^{I_2}\cdots \Psi^{I_k]})\approx \Tr\left(\rho_{\text{eff}}C^{I_1}C^{I_2}\cdots C^{I_k}\right).
\end{equation}

Using Eq.~\eqref{state}, this problem becomes a weight-by-weight matching of the fermionic correlation coefficients of the EFT state and the fundamental state for all fermionic strings with weight $k\le 2n$. More explicitly, we seek a fundamental state $\rho_F$ that minimizes the mismatch,
\begin{equation}
    \min_{\rho_F}\sum_{I_1I_2\cdots I_k}\left|\Tr(\rho_{\text{eff}}C^{[I_1}C^{I_2}\cdots C^{I_k]})-V^{[I_1}_{i_1}V^{I_2}_{i_2}\cdots V^{I_k]}_{i_k}\Tr(\rho_F C^{[i_1}C^{i_2}\cdots C^{i_k]})\right|^2, \qquad k\le 2n.
\end{equation}

However, the number of independent $k$-point functions available in the fundamental state is much smaller than the number of independent EFT $k$-point coefficients. As a result, the minimized cost generally remains large; see Appendix~\ref{app:compression} for details. This obstruction originates from the single-particle overlap assumption in Eq.~\eqref{eq:single}, and should motivate us to consider more general overlapping-operator models.

\subsection{General construction}
Another question is whether the structures we observe in the previous section is specific to fermionic operators.
We argue that these structures persist by providing a more general way of embedding the   overlapping degrees of freedom into fundamental degrees of freedom, which includes Pauli operators. In particular, we recover a factorized form analogous to Eq.~\eqref{eq:recon} for observed effective state. 

Let $\mathcal{H}_d$ be the $d$-dimensional fundamental Hilbert space.  
Choose a Hilbert--Schmidt orthonormal basis of operators

\[
\{C_a\}_{a=0}^{d^2-1}\subset \mathrm{Herm}(\mathcal{H}_d),
\qquad 
\mathrm{Tr}(C_a C_b)=\delta_{ab}, \qquad C_0=I/d
\]

The basis is complete in the sense that
\begin{equation}
\sum_{a=0}^{d^2-1} (C_a)_{ij}\,(C_a^\ast)_{kl}
= \delta_{ik}\,\delta_{jl}.
\label{eq:Ca-complete}
\end{equation}

To construct the general overlapping basis, we  build an over-complete family of operators on $\mathcal{H}_d$, related to $C_a$'s by
\begin{equation}
B_I := \sum_{a=1}^{d^2-1} v_{Ia}\, C_a,
\qquad I=1,\dots,K.
\label{eq:CI-def}
\end{equation}
where $v_{Ia}$'s are $K$ number of  i.i.d.\ random unit vectors drawn from $\mathbb{R}^{d^2}$.

Since the $\vec v_I$ are random and isotropic, their empirical Gram matrix is
approximately proportional to the identity:
\begin{equation}
\sum_{I=1}^{K} v_{Ia}\,v_{Ib}^\ast \;\approx\; f\,\delta_{ab},
\qquad 
f \approx \frac{K}{d^2}.
\label{eq:v-gram}
\end{equation}

Plugging \eqref{eq:CI-def} into the completeness relation gives an ``over-completeness'' condition:
\begin{align}
\sum_{I=1}^{K} (B_I)_{ij}\,(B_I^\ast)_{kl}
&=
\sum_{I=1}^{K}\sum_{a,b=1}^{d^2-1}
v_{Ia}\,v_{Ib}^\ast\,(C_a)_{ij}\,(C_b^\ast)_{kl} \notag \\
&\approx
f\sum_{a=1}^{d^2-1} (C_a)_{ij}\,(C_a^\ast)_{kl}
=
(1-d^{-2})f\,\delta_{ik}\,\delta_{jl}
\label{eq:CI-overcomplete}
\end{align}

where we take $K=D^2-1$, corresponding to a larger Hilbert space $\mathcal{H}_D$. Choose an orthonormal basis of traceless Hermitian operators of dimension $D$, 
\[
\{P_I\}_{I=0}^{D^2-1}\subset \mathrm{Herm}_0(\mathcal{H}_D),
\qquad 
\mathrm{Tr}(P_I P_J)=\delta_{IJ}, \qquad P_0=I/D.
\]

 Given a density matrix $\rho$ on $\mathcal{H}_d$,  define the linear reconstruction
on $\mathcal{H}_D$ using the measured expectation values $\mathrm{Tr}(\rho B_I)$:
\begin{equation}
\rho_{\mathrm{rec}}
:= P_0 + \sum_{I=1}^{D^2-1} \mathrm{Tr}(\rho\,B_I)\,P_I.
\label{eq:rho-rec}
\end{equation}
where $P_0 := \frac{1}{D}\,\mathbf{1}_D$ is the identity component. 

Using \eqref{eq:CI-def}, this can be written explicitly in terms of the original
basis coefficients $\mathrm{Tr}(\rho C_a)$:
\begin{equation}
\rho_{\mathrm{rec}}
= P_0 + \sum_{I=1}^{D^2-1}\sum_{a=1}^{d^2-1}
v_{Ia}\,\mathrm{Tr}(\rho\,C_a)\,P_I.
\label{eq:rho-rec-expanded}
\end{equation}

Now we show this state reconstructed from measurement outcomes spans only a small subspace in $\mathcal{H}_D$. 

Consider a linear transformation $\mathcal{E}$ on operators over $\mathcal{H}_D$
which fixes the identity $\mathbb I_D$ and acts as an orthogonal rotation on the traceless
subspace, then

\begin{equation}
\mathcal{E}(P_0)=P_0,
\qquad 
\mathcal{E}(P_I)=\sum_{J=1}^{D^2-1} O_{IJ} P_J,
\qquad 
O\in SO(D^2-1).
\label{eq:E-action}
\end{equation}
Applying $\mathcal{E}$ to \eqref{eq:rho-rec-expanded} gives
\begin{align}
\mathcal{E}(\rho_{\mathrm{rec}})
&= P_0 + \sum_{I=1}^{D^2-1}\sum_{a=1}^{d^2-1}
v_{Ia}\,\Tr(\rho\,C_a)\,\mathcal{E}(P_I) \notag\\
&= P_0 + \sum_{I=1}^{D^2-1}\sum_{a=1}^{d^2-1}\sum_{J=1}^{D^2-1}
v_{Ia}\,O_{IJ}\,\Tr(\rho\,C_a)\,P_J \notag\\
&= P_0 + \sum_{J=1}^{D^2-1}\sum_{a=1}^{d^2-1}
( O^{\mathsf T} v )_{Ja}\,\Tr(\rho\,C_a)\,P_J .
\label{eq:E-rho-rec}
\end{align}

Take the singular value decomposition of $v$:
\begin{equation}
v = S\,\Lambda\,Q^{\mathsf T},
\qquad 
S\in SO(D^2-1),\quad Q\in SO(d^2-1),
% \quad \Lambda=\mathrm{diag}(\lambda_1,\dots,\lambda_{d^2}),
\label{eq:svd}
\end{equation}
where $\Lambda$ is a $(D^2-1)\times (d^2-1)$ rectangular diagonal matrix (with eivenvalues $\{\lambda_a\}_{a=1}^{d^2}$ and zeros
below the first $d^2$ rows).

Substitute \eqref{eq:svd} into \eqref{eq:E-rho-rec}:
\begin{align}
\mathcal{E}(\rho_{\mathrm{rec}})
&= P_0 + \sum_{J=1}^{D^2-1}\sum_{a=1}^{d^2-1}
\bigl(O^{\mathsf T}S\Lambda Q^{\mathsf T}\bigr)_{Ja}\,\Tr(\rho\,C_a)\,P_J \notag\\
&= P_0 + \sum_{J=1}^{D^2-1}\sum_{a,b=1}^{d^2-1}
\bigl(O^{\mathsf T}S\Lambda\bigr)_{Ja}\,Q_{ba}\,\Tr(\rho\,C_b)\,P_J .
\label{eq:E-rho-rec-svd}
\end{align}

For the following choice of the rotation matrix $O$,
\begin{equation}
O = S \quad\Longrightarrow\quad O^{\mathsf T}S = \mathbf{1}_{D^2-1},
\label{eq:choose-O}
\end{equation}
 \eqref{eq:E-rho-rec-svd} simplifies to
\begin{align}
\mathcal{E}(\rho_{\mathrm{rec}})
&= P_0 + \sum_{J=1}^{D^2-1}\sum_{a,b=1}^{d^2-1}
\Lambda_{Ja}\,Q_{ba}\,\Tr(\rho\,C_b)\,P_J \notag\\
&= P_0 + \sum_{a,b=1}^{d^2-1}
\lambda_a\,Q_{ba}\,\Tr(\rho\,C_b)\,P_a .
\label{eq:E-rho-rec-aligned}
\end{align}
Then we identify the rotated operator basis on the $d$-dimensional space
\begin{equation}
C'_a := \sum_{b=1}^{d^2-1} Q_{ba}\,C_b
\qquad\Longrightarrow\qquad
\Tr(\rho\,C'_a)=\sum_{b=1}^{d^2-1} Q_{ba}\,\Tr(\rho\,C_b).
\label{eq:Cprime}
\end{equation}
Hence
\begin{equation}
\mathcal{E}(\rho_{\mathrm{rec}})
= P_0 + \sum_{a=1}^{d^2-1}\lambda_a\,\Tr(\rho\,C'_a)\,P_a .
\label{eq:E-final}
\end{equation}

If one chooses the large-space operator basis $\{P_I\}_{I=1}^{D^2-1}$ so that its subset
$\{P_a\}_{a=1}^{d^2-1}$ act only on a $d$-dimensional factor in a tensor
decomposition $\mathcal{H}_D \cong \mathcal{H}_d\otimes \mathcal{H}_{D/d}$
(assuming $d\mid D$), then \eqref{eq:E-final} can be viewed as producing an
embedded state of the form
\[
\mathcal{E}(\rho_{\mathrm{rec}})\;\equiv\; \rho'_{d\times d}\ \otimes\
\frac{1}{D/d}\mathbf{1}_{D/d},
\]
for some $\rho'_{d\times d}$ determined by the coefficients
$\lambda_a\,\Tr(\rho\,C'_a)$. 

Therefore, a computationally powerful observer who collects enough effective qudits and sufficiently many state samples can still infer that the states actually lie in a lower-dimensional subspace than their effective description suggests.

\section{Discussion}
\label{sec:disc}
In this work, we have introduced the idea of overlapping qubits and fermions to the context of the black hole information problem. These quantum information constructions can provide an operable mathematical framework for approximate tensor factorization and approximately local quantum field theories. In addition, they offer a physically intuitive way to understand how information leaves a black hole from a Hilbert space perspective by realizing unitarity and non-violent non-locality while reproducing a Page curve in the quasi-local description of random fermionic Gaussian states. Spectral analyses of the general operator algebra reconstruction are still needed to verify whether the purity observation in the fermionic Page curve can be reproduced in other physically relevant systems. 

Given the immense freedom one has in constructing different overlapping compression maps, we anticipate a rich landscape of theories, some of which may be useful candidates that can extend the idea of approximate/non-violent non-locality to quantum field theories. 
An immediate extension is to identify more general overlapping maps that can approximately reproduce low-point functions of local quantum field theories, thus providing a direct path for reconciling with the no-drama postulate in the firewall paradox. 
 The same approximate factorization that avoided the violation of the no-cloning theorem indicates that the conventional strong subadditivity arguments~\cite{Almheiri_2013,Mathur_2009} for the emergence of firewalls would not apply either. The tradeoff is that now locality is only approximate in the algebraic sense, where operators on distinct overlapping qubits commute up to $\epsilon$ corrections.

The breakdown of factorizability  in the quasi-local picture also means that familiar concepts and no-go theorems that are built upon well-defined subsystems (or subalgebras) of the Hilbert space such as entanglement  and entropy inequalities need to be re-examined with care. In particular, a proper choice of the overlapping map that can reproduce the low-point correlators of local effective field theories would be necessary to truly preserve the no-drama postulate in a physical sense instead of purely relying on the appearance of approximate locality in operator algebra only. A ``Schr\"odinger picture" analogue of overlapping qubits, would also allow us to refine the current toy model, by introducing the entanglement between the overlapping ``Hawking modes". We will leave these open questions for future work.

In constructing the overlapping operators, non-isometric maps are used \cite{Akers:2022qdl,Cao_2025,Antonini:2024yif,noniso1,noniso2,noniso3,noniso4}. This thus also calls into question the nature of the effective dynamics if the fundamental dynamics are assumed to be unitary. Formally, the effective dynamics should not be unitary, but the breaking of unitarity may only become more apparent as the operators become higher weight or under operations where the non-orthogonal corrections become dominant. Though this may be surprising in our current quantum field theory intuition, it need not contradict our physical observations and may be an interesting avenue for future explorations of more foundational questions of quantum mechanics, such as quantum-to-classical transition. 

Another key aspect of overlapping operators is that they are highly redundant since they can be restricted down to a smaller subset of truly independent operators in a multitude of different ways. This is reminiscent of gauge redundancy and suggests that there might be an underlying gauge in overlapping qubit constructions. This is exactly where the presence of non-locality would be attributed to in a conventional quantum field theory, but since there is not a QFT analogue of overlapping qubits, this connection is not well understood. However, at the classical level the presence of overlap in a regularized scalar field theory has been identified~\cite{friedrich2026areascalingdynamicaldegrees}. The interpretation there might be different, but understanding how these results are generalized at the quantum (and eventually infinite-dimensional) level will hopefully also illuminate the connections with gauge theory.

 On a more speculative note, we would like to point out two observed similarities between certain features observed in the overlapping qubits toy model and in the replica wormholes of the gravitational path integral. In~\cite{Penington:2019kki} it was argued that the gravitational path integral performs an implicit averaging over an ensemble of theories such that different black hole microstates $\ket{\psi_i},\ket{\psi_j}$ are orthogonal on average while still having an overlap of non-zero magnitude, which is exponentially suppressed by the black hole entropy, i.e.
\begin{equation}
    \overline{\braket{\psi_i}{\psi_j}}=0\quad,\quad \overline{|\braket{\psi_i}{\psi_j}|^2}\sim e^{-S}\quad,\quad i\neq j.
\end{equation}
The average was speculated to be over microscopic quantities $R_{ij}$ such that:
\begin{equation}
    \braket{\psi_i}{\psi_j}=\delta_{ij}+e^{-\frac{S_0}{2}}R_{ij}.
\end{equation}
First, approximately orthogonal objects appear in both constructions. Furthermore, they are accompanied by random variables in the natural product defined between them and an implicit averaging. A similar averaging can be seen in the antisymmetrization of the overlapping fermionic operators, where one averages over possible non-equivalent orderings of the operators.

It is still unclear if there is a connection between overlapping qubits and the gravitational path integral approach, but a more concrete link between the two could provide a guide for more realistic overlapping qubit constructions. It will also be interesting to understand any potential connection between the factorized identity block and black hole islands.
\section*{Acknowledgments}

We thank Ning Bao, Aidan Chatwin-Davies, Wissam Chemissany, Oliver Friedrich, Jason Pollack, John Preskill, and Zolt\'an Zimbor\'as for helpful comments and discussions. We also acknowledge the use of ChatGPT and Claude Opus 4.8 for beautifying plot graphics.

\appendix
\section{Negativity correction algorithm}
\label{sec:proofs}
In this appendix we will provide proofs for Lemma~\ref{theorem1} and Lemma~\ref{theorem2}.
Let us begin by briefly reviewing the algorithm of~\cite{Smolin_2012}, which finds the closest valid density matrix under the Frobenius norm
to the experimentally noisy matrix we are studying.
Consider the $d\times d$ matrix $\sigma$ with eigenvalues $\{\lambda_i\}$ such that $\lambda_i\geq \lambda_{i+1}$ and associated eigenvectors $\ket{\lambda_i}$. This algorithm takes $\sigma$ as an input and outputs the density matrix
\begin{equation}
    \rho:=\sum_{i=1}^d \lambda'_i \ketbra{\lambda_i}
\end{equation}
saturating
$$
\min_{\rho} \Vert\rho-\sigma\Vert_2^2,
$$
where the minimization is over density matrices. 

The algorithm is as follows:

\begin{enumerate}
    \item Initialize $a=0$ and start from $i=d$.
    \item If $\lambda_i+\frac{a}{i}\leq 0$, set $\lambda'_i=0$, update $a\rightarrow a+\lambda_i$ and $i\rightarrow i-1$. Repeat step 2.
    \item If $\lambda_i+\frac{a}{i}>0$ for all $j\leq i$ set $\lambda'_j=\lambda_j+\frac{a}{i}$.
    \item The density matrix minimizing the 
    % trace distance 
    Frobenius distance $\Vert\rho-\sigma\Vert_2^2$ is
    \begin{equation}
        \rho:=\sum_{i=1}^d \lambda'_i\ketbra{\lambda_i}
    \end{equation}
\end{enumerate}

Essentially, the algorithm redistributes the negativity equally to the remaining positive eigenvalues such that eventually only non-negative eigenvalues remain.

We are ready to begin from the proof of~\ref{theorem1}. 
It will prove convenient to prove the following lemma first:
\begin{lemma}
    Continuing the iterative process of the negativity correction algorithm past the point where the check fails would never cause the check to succeed again.\label{Lemma}
\end{lemma}
\begin{proof}
We want to show that
\begin{equation}
    \lambda_i+\frac{a}{i}>0\Rightarrow \lambda_{i-1}+\frac{a+\lambda_i}{i-1}>0.
\end{equation}
Then the statement is trivially proven by induction, since we have built into the premise that the check succeeds at some point.
However, the above relation follows from
\begin{align}
    \lambda_{i-1}+\frac{a+\lambda_i}{i-1}&=\lambda_{i-1}+\frac{a}{i}+\frac{\lambda_i+\frac{a}{i}}{i-1}\\
    &\geq\lambda_{i}+\frac{a}{i}+\frac{\lambda_i+\frac{a}{i}}{i-1}\\
    &>0,
\end{align}
where in the second line we used $\lambda_i\geq \lambda_{i+1}$ and in the last line we used 
\begin{align}
    \lambda_i+\frac{a}{i}&>0.\quad\Box
\end{align}
\end{proof}
Now we are ready to prove Lemma~\ref{theorem1}.
\begin{proof}
Note that the corrected density matrix $\rho$ is pure iff it has only one non-zero eigenvalue.
Therefore it suffices to show that the check fails at $\lambda_2$. 
Lemma~\ref{Lemma} ensures that if the check fails here it can't have succeeded earlier.
In other words we need to show that
\begin{equation}
    \lambda_2+\frac{a}{2}\leq 0.
\end{equation}
Notice that at this step we have
\begin{equation}
    a=\sum_{i=3}^d \lambda_i=\Tr(\sigma)-\lambda_1-\lambda_2.\label{next}
\end{equation}
Then
\begin{align}
    \lambda_2+\frac{a}{2}&=\lambda_2+\frac{\Tr(\sigma)-\lambda_1-\lambda_2}{2}\nonumber\\
    &=\lambda_2+\frac{1-\lambda_1-\lambda_2}{2}\nonumber\\
    &=\frac{1-(\lambda_1-\lambda_2)}{2}\nonumber\\
    &\leq 0,
\end{align}
where in the last line we used $\lambda_1-\lambda_2>1$, which is given to us by the premise.
For completeness, it should be said that the left hand side of~\ref{next} is the check at the final step, which succeeds. 
Therefore, the state is indeed corrected to a pure state. $\Box$
\end{proof}
Finally, we will prove Lemma~\ref{theorem2}.
\begin{proof}
    Let $\tilde{\lambda}_c$ be the eigenvalue where the check of the correction algorithm applied to $\tilde{\sigma}$ succeeds for the last time.
    Essentially, the label $c$ in addition to the spectrum of $\tilde{\sigma}$ uniquely determines the correction algorithm.
    The fact that the check succeeds at $\tilde{\lambda}_c$ means that
    \begin{align}
        0&\geq c\tilde{\lambda}_c+\tilde{a}\nonumber\\
        &=c\tilde{\lambda}_c+\Tr(\tilde{\sigma})-\sum_{i={1}}^c \tilde{\lambda_i}\nonumber\\
        &=(c-1)\tilde{\lambda}_c+1-\sum_{i={1}}^{c-1} \tilde{\lambda_i}.\label{a11}
    \end{align}
    Similarly we know that
    \begin{align}
        0&<(c-1)\tilde{\lambda}_{c-1}+{\tilde{a}+\tilde{\lambda}_c}\nonumber\\
        &=(c-1)\tilde{\lambda}_{c-1}+\Tr(\tilde{\sigma})-\sum_{i={1}}^{c-1} \tilde{\lambda_i}\nonumber\\
        &=(c-1)\tilde{\lambda}_{c-1}+1-\sum_{i={1}}^{c-1} \tilde{\lambda_i}.\label{a12}
    \end{align}
    This fixes the spectrum of the corrected density matrix $\tilde{\rho}':=\Phi({\tilde \sigma})$ to be
\begin{equation}
    {\tilde \lambda_i}'=\begin{cases}
        0&i\geq c\\
        \tilde{\lambda}_i+\frac{\tilde{a}+\tilde{\lambda}_c}{c-1}& i<c
    \end{cases}.
\end{equation}    
Now we move to the matrix $\mu$. 
The tensoring with the identity rescales all eigenvalues by $\frac{1}{d}$ and increases their multiplicity to $d$.
At the last of the rescaled $\lambda_c=\frac{{\tilde \lambda}_c}{d}$ we have
\begin{align}
    \lambda_c+\frac{a}{d(c-1)+1}&=\lambda_c+\frac{1-d\sum_{i=1}^{c-1}\lambda_i-\lambda_c}{d(c-1)+1}\nonumber\\
    &=\frac{1+d\left((c-1)\lambda_c-\sum_{i=1}^{c-1}\lambda_i\right)}{d(c-1)+1}\nonumber\\
    &=\frac{(c-1)\tilde{\lambda}_c+1-\sum_{i=1}^{c-1}\tilde{\lambda}_i}{d(c-1)+1}\nonumber\\
    &\leq 0,
\end{align}
where in the last step we used~\ref{a11}.
In the next step we have
\begin{align}
    \lambda_{c-1}+\frac{a+\lambda_c}{d(c-1)}&=\lambda_{c-1}+\frac{1-d\sum_{i=1}^{c-1}\lambda_i}{d(c-1)}\nonumber\\
    &=\frac{(c-1)\tilde{\lambda}_{c-1}+1-\sum_{i=1}^{c-1}\tilde{\lambda}_i}{d(c-1)}\nonumber\\
    &>0,
\end{align}
where in the last step we used~\ref{a12}.
Therefore, the check fails and the eigenvalues of $\mu'$ are (each with their $d$-fold multiplicity)
\begin{align}
    \lambda_i'&=\begin{cases}
        0&i\geq c\\
        \lambda_i+\frac{a+\lambda_c}{d(c-1)}&i<c
    \end{cases}\nonumber\\
    &=\begin{cases}
        0&i\geq c\\
        \frac{1}{d}\left(\tilde{\lambda}_i+\frac{1-d\sum_{i=1}^{c-1}\lambda_i}{c-1}\right)&i<c
    \end{cases}\nonumber\\\
    &=\begin{cases}
        0&i\geq c\\
        \frac{1}{d}\left(\tilde{\lambda}_i+\frac{\tilde{a}+\tilde{\lambda}_c}{c-1}\right)&i<c
    \end{cases}\nonumber\\
    &=\frac{1}{d}\tilde{\lambda_i}'.
\end{align}
But since the correction algorithm leaves the eigenvectors unchanged and all eigenvalues of $\tilde{\rho}'$ have been rescaled by $\frac{1}{d}$ with their multiplicity increased to $d$ this means that
\begin{equation}
    \Phi(\mu):=\mu'=\tilde{\rho}'\otimes \frac{1}{d}{\mathbb I}_{d\times d}.\quad\Box
\end{equation}
\end{proof}
\section{Singular values of V}
\label{sec:singvalv}
 In the main text we claimed that the singular values of the matrix containing the overlapping vectors are always greater than one, which led to the enlargement/shrinking of the spectrum for our density matrices under embedding/compression, which was a key feature for our results for Gaussian fermionic states. In this appendix we will prove this claim.\\

The starting point in the constructions we are using is a random orthogonal matrix $A$ from which we project out some rows. So first, consider the matrix:
\begin{equation}
    W_{m\times n}:=P_{m\times n}A_{n\times n}\quad,\quad P_{m\times n}:=\begin{pmatrix}
        {\mathbb I}_{m\times m}&0_{m\times (n-m)}
    \end{pmatrix}
\end{equation}
The projection is random in our case, but without loss of generality for our result, let's imagine we are keeping only the first $m$ rows. Now the singular values are determined by the eigenvalues of:
\begin{equation}
    WW^T=PAA^TP^T=PP^T={\mathbb I}_{m\times m}
\end{equation}
Thus, the singular values of $W$ at this point are all 1, since singular values are positive by convention.\\

However, we also have to renormalize the columns of the matrix so that they have unit length. The columns prior to the projection had norm equal to 1, so the projection makes the norm of the overlapping vectors smaller than 1 (at best it leaves it equal to 1, but for that to happen we would need to project out only zero entries randomly from a random matrix, which is for all practical purposes impossible). Thus, since the renormalizing factor is:
\begin{equation}
    D_I:=\frac{1}{\sqrt{\sum_{j}W_{jI}W_{jI}}}\geq 1
\end{equation}
We can think of this action as a right multiplication with a diagonal matrix whose entries are all greater (or equal to 1). Thus:
\begin{equation}
    V=WD=PAD
\end{equation}
However:
\begin{equation}
    D^2\succeq {\mathbb I}_{n\times n}\Leftrightarrow AD^2A^T \succeq {\mathbb I}_{n\times n}\Leftrightarrow x^T (AD^2 A^T -{\mathbb I}_{n\times n})x\geq 0 \quad,\quad  \forall x\in {\mathbb R}^n
\end{equation}
This means, that the above inequality also holds for the $x$ that are of the form:
\begin{equation}
    x=P^Ty \quad,\quad \forall y\in{\mathbb R}^m
\end{equation}
Therefore:
\begin{equation}
    x^T(AD^2A^T-{\mathbb I}_{n\times n})x\geq 0 \Rightarrow y^T(PAD^2A^TP^T-{\mathbb I}_{m\times m})y\quad,\quad \forall y \in{\mathbb R}^m
\end{equation}
Thus:
\begin{equation}
    VV^T\succeq {\mathbb I}_{m\times m}\Leftrightarrow U \Lambda U^T\succeq {\mathbb I}_{m\times m}\Leftrightarrow \Lambda \succeq {\mathbb I}_{m\times m} 
\end{equation}
Which completes the proof that the singular values of $V$ are always (for practical purposes) greater than 1.

\section{Proof of Lemma~\ref{lm:william} and Theorem~\ref{lm:pure}}\label{app:pure}
\subsection*{Proof of Lemma~\ref{lm:william}}
\begin{proof}
The covariance matrix $\tilde\Gamma$ of our pseudo-state is
\begin{equation}
    \tilde\Gamma_{IJ}=\frac{i}{2}\Tr(\tilde{\rho}[C_{I},C_{J}])=\sum_{ij}V_{iI}V_{jJ}\Gamma_{ij}^F,
\end{equation}
where $V$ is $N\times n$ matrix with the SVD decomposition 
\begin{align}
    V=S\Sigma O^T.
\end{align}
Therefore the orthogonal matrix $S\in SO(N)$ can transform the pseudo-state covariance matrix to 
\begin{align}
    S^T\tilde\Gamma S =\left(\sqrt{\Lambda} O^T\Gamma^F O\sqrt{\Lambda}\right)\oplus \mathbf{0},
\end{align}
where $\Lambda$ is the $n\times n$ diagonal matrix constructed from singular values of $V$.

The covariance matrix $\Gamma^F$ can be brought to the canonical form by a rotation $R$.
Let the Williamson eigenvalues  of the fundamental covariance matrix be $\nu_i$. Then we have:
\begin{equation}
    \Gamma^F_{ij}=\left[R\left(\bigoplus_{i=1}^n \begin{pmatrix}
        0&\nu_i\\-\nu_i&0
    \end{pmatrix} \right)R^T\right]_{ij}
\end{equation}
Since the state $\rho_F$ in the fundamental space is pure, we expect $|\nu_i|=1$.

Now we sort the singular values $\sigma_I$ to descending order and  bound the non-zero singular values of $\tilde\Gamma$ using the singular value inequility, $\sigma_i(AB)\geq \sigma_i(A)\sigma_{\min}(B)$~\cite{doi:10.1137/0732088}, which leads to

\begin{equation}
    \begin{aligned}
    \sigma_I\left(\tilde\Gamma\right)=&\sigma_I\left(\sqrt{\Lambda}\Gamma^F\sqrt{\Lambda}\right)\\
    \geq& \sigma_I\left(\sqrt{\Lambda}\Gamma^F\right)\sigma_{\min}\left(\sqrt{\Lambda}\right)\\
    \geq & \sigma_I\left(\Gamma^F\right)\sigma_{\min}\left(\sqrt{\Lambda}\right)^2, \qquad I=1,2,\cdots n.
\end{aligned}
\end{equation}
Since the singular values of $\tilde\Gamma$ are equal to the absolute value of its Williamson values, we proved the following inequility for the non-zero Williamson values $\nu_I$ of $\tilde\Gamma$
\begin{align}
    |\tilde\nu_I|\geq \lambda_{\min}|\nu_I|\geq \lambda_{\min}\geq 1, \qquad I=1,2,3\cdots n,  
\end{align}
where we have used the bound of $\lambda_{\min}$ derived in Appendix~\ref{sec:singvalv}. 
\end{proof}

\subsection*{Proof of Theorem~\ref{lm:pure}}
\begin{proof}
An implication of~\ref{eq:recon} is that since the singular values of the matrix $V$ are always greater than 1 (you can see appendix~\ref{sec:singvalv} for the proof), embedding the overlapping fermions in a higher-dimensional Hilbert space, enlarges the spectrum of the operator. Since we have a pure state in the fundamental Hilbert space by premise, this justifies us in parameterizing
\begin{equation}
    \nu_i=1+2\delta \nu_i\quad,\quad \delta \nu_i>0,
\end{equation}
which is also supported by numerics. {Without loss of generality, the $\nu_i$ are also labeled in non-increasing order.} Thus:
\begin{equation}
    \rho=\bigotimes_{i=1}^n\begin{pmatrix}
        1+\delta \nu_i&0\\0&-\delta \nu_i
    \end{pmatrix}
\end{equation}
Now we are ready to show that Lemma~\ref{theorem1} applies by focusing on the two greatest eigenvalues. From the structure of the density matrix it is clear that these are
\begin{equation}
    \nu_1=\prod_{i=1}^n(1+\delta \nu_i)\quad,\quad \nu_2=\delta \nu_1\delta \nu_2\prod_{i=3}^n(1+\delta \nu_i)
\end{equation}
Then, we see that
\begin{align}
    \nu_1-\nu_2=\left(1+\delta \nu_1+\delta \nu_2\right)\prod_{i=3}^n(1+\delta \nu_i)>1
\end{align}
and indeed Lemma~\ref{theorem1} applies. Thus, the correction algorithm always gives us back a pure state.
% \TK{The last sentence shouldn't be here in his format}
\end{proof}

\section{Compression map}\label{app:compression}
Define the embedding map $\mathcal{N}_{E}$ as in Eq.~\eqref{state} : 
\begin{equation}
    \mathcal{N}_{E}(\rho_F)=\frac{1}{2^N}\sum_{k=0}^{2N}\sum_{I_1I_2\cdots I_k} (i)^{k(k-1)}\Tr(\rho_F \Psi_{[I_1}\Psi_{I_2}\cdots \Psi_{I_k]})C_{[I_1}C_{I_2}\cdots C_{I_k]}. 
\end{equation}

\begin{equation}\label{eq:state_approx}
    \mathcal{N}_{E}(\rho_f)\approx \rho_{\text{eff}},
\end{equation}
The condition Eq.~\eqref{eq:state_approx} should be understood as a constraint on the low weight correlation functions of $\rho_F$
\begin{equation}
    \Tr(\rho_F \Psi_{[I_1}\Psi_{I_2}\cdots \Psi_{I_k]})\approx \Tr\left(\rho_{\text{eff}}C_{I_1}C_{I_2}\cdots C_{I_k}\right).
\end{equation}

This becomes a matching of the fermionic coefficients between the effective state and fundamental state weight by weight as long as the weight $k$ of fermionic string is below $2n$. More specifically, we require $\rho_F$ to minimize the following difference:
\begin{equation}
    \min_{\rho_F}\sum_{I_1I_2\cdots I_k}\left|\Tr(\rho_{\text{eff}}C_{[I_1}C_{I_2}\cdots C_{I_k]})-V_{i_1I_1}V_{i_2I_2}\cdots V_{i_kI_k}\Tr(\rho_F C_{[i_1}C_{i_2}\cdots C_{i_k]})\right|^2, \qquad k\le 2n.
\end{equation}

This is a least square fitting problem at each weight $k$. Define $(VV^T)_{ij}:=\sum_I V_{iI}V_{jI}$, we could write down the solution as 
\begin{equation}
\begin{aligned}
    &\Tr(\rho_F C^{i_1}C^{i_2}\cdots C^{i_k})\\
    =&\sum_{j_1j_2\cdots j_k}\sum_{I_1I_2\cdots I_k}(VV^T)^{-1}_{i_1j_1}(VV^T)^{-1}_{i_2j_2}\cdots (VV^T)^{-1}_{i_kj_k}V_{j_1I_1}V_{j_2I_2}\cdots V_{j_kI_k}\,\Tr(\rho_{\text{eff}}C_{I_1}C_{I_2}\cdots C_{I_k}). 
\end{aligned}
\end{equation}

This defines the compression map 
\begin{equation}
     \rho_F = \mathcal{N}_C(\rho_{\text{eff}}). 
\end{equation}

If the state $\rho_{\text{eff}}$ is Gaussian, the correlation functions factorizes, and one can show that the solution $\rho_{F}$ of the fitting problem is also a Gaussian state. For instance, the fourth moment is
\begin{equation}
\begin{aligned}
    &\Tr(\rho_F C_{i_1}C_{i_2} C_{i_3}C_{i_4})\\
    =&(VV^T)^{-1}_{i_1j_1}(VV^T)^{-1}_{i_2j_2}V_{j_1I_1}V_{j_2I_2}\left({\Gamma^{I_1I_2}\Gamma^{I_3I_4}+\cdots}\right)(VV^T)^{-1}_{i_3j_3}(VV^T)^{-1}_{i_4j_4}V_{j_3I_3}V_{j_4I_4}\\
    =&\Gamma_{c}^{i_1i_2}\Gamma_{c}^{i_3i_4}+\cdots
\end{aligned}
\end{equation}
where \begin{equation}
\begin{aligned}
    \Gamma_{c}^{i_1i_2}\equiv&\sum_{j_1j_2I_1I_2}(VV^T)^{-1}_{i_1j_1}(VV^T)^{-1}_{i_2j_2}V_{j_1I_1}V_{j_2I_2}\, \Gamma^{I_1I_2}\\
    = & [(VV^T)^{-1}V\,\Gamma\, V^T(VV^T)^{-1}]^{i_1i_2}, 
\end{aligned}
\end{equation}
is the compressed covariant matrix. Therefore the Gaussian state $\rho_{\text{eff}}$ is compressed to a Gaussian state $\rho_{F}$ whose covariant matrix is solved by the least-square fitting problem with the covariant matrix of $\rho_{\text{eff}}$. 

\subsection*{Overlap of compressed states}
We pick two states $\rho_{\text{eff}}$ and $\rho'_{\text{eff}}$ randomly from the set 
\begin{equation}
    \rho=U\prod_{j=1}^{N}\frac{1+(-1)^{s_j}i\psi_{2j-1}\psi_{2j}}{2}U^{\dagger}
\end{equation}
with $s_j\in \{0,1\}$ being random variable, and $U$ being some fixed unitary.  Then we compute the fidelity of the compressed states 
\begin{equation}
    \mathcal{F}\left({\mathcal{N}_C(\rho_{\text{eff}}), \mathcal{N}_C(\rho'_{\text{eff}})}\right),
\end{equation}
using the expression in covariant matrix \cite{Swingle_2019},
\[
\mathcal{F}(\rho_1,\rho_2)
=
\left(
\det\!\left[\frac{\mathbb{I}-\Gamma_1\Gamma_2}{2}\right]\,
\det\!\left[
\mathbb{I}+\left(\mathbb{I}+\left((\Gamma_1+\Gamma_2)(\mathbb{I}-\Gamma_1\Gamma_2)^{-1}\right)^2\right)^{1/2}
\right]
\right)^{1/4}.
\]

We find that the fidelity decays exponentially with the Hamming distance between the two bit-strings:
\begin{equation}
    \mathcal{F}\left({\mathcal{N}_C(\rho_{\text{eff}}), \mathcal{N}_C(\rho'_{\text{eff}})}\right)\sim e^{-\alpha(N,n)|s-s'|},
\end{equation}
where $\lvert s-s'\rvert$ denotes the Hamming distance and $\alpha(N,n)$ is a decay rate that depends on $N$ and $n$.

For fixed compression ratio $N/n$, we plot $\log \mathcal{F}$ as a function of $\lvert s-s'\rvert$ in Fig.~\ref{fig:compression1}. To determine how the decay rate depends on the compression ratio, we calculate
\begin{equation}
\log\alpha
\equiv
\log\left(
-\frac{\log \mathcal{F}}{\lvert s-s'\rvert}
\right).
\end{equation}
 In Fig.~\ref{fig:compression2}, we plot $\log\alpha$ against $\log(N/n)$ at fixed Hamming distance and fit the data to
\begin{equation}
\log\alpha=\beta\log(N/n)+c.
\end{equation}
The fitted slope is $\beta\simeq-4.02$, implying
\begin{equation}
\alpha(N,n)
=\mathcal{O}\left[\left(\frac{n}{N}\right)^{|\beta|}\right].
\end{equation}

For two typical bit strings of length $N$, the Hamming distance is of order $N/2$. Their fidelity therefore scales as
\begin{equation}
    \mathcal{F}\sim \exp\left[
-O\left(
n^{|\beta|}N^{1-|\beta|}
\right)
\right].
\end{equation}

Consequently, the compressed states become asymptotically orthogonal provided that $n\geq c N^{1-|\beta|^{-1}}$. 

\begin{figure}[t]
    \centering
    \begin{subfigure}{0.48\textwidth}
        \centering
        \includegraphics[width=\linewidth]{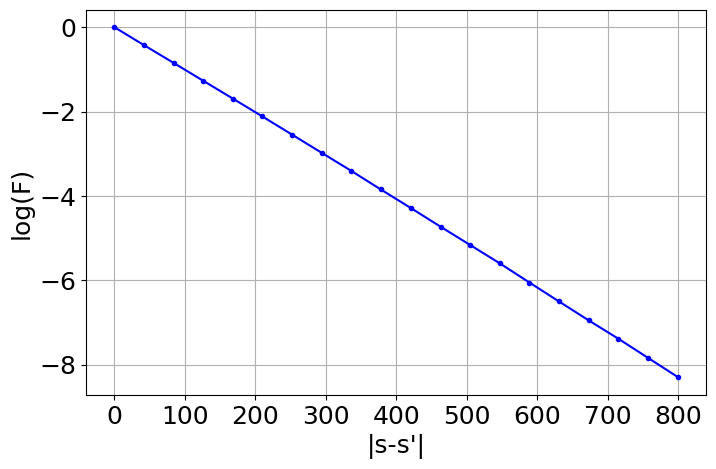}
        \caption{}
        \label{fig:compression1}
    \end{subfigure}
    \hfill
    \begin{subfigure}{0.48\textwidth}
        \centering
        \includegraphics[width=\linewidth]{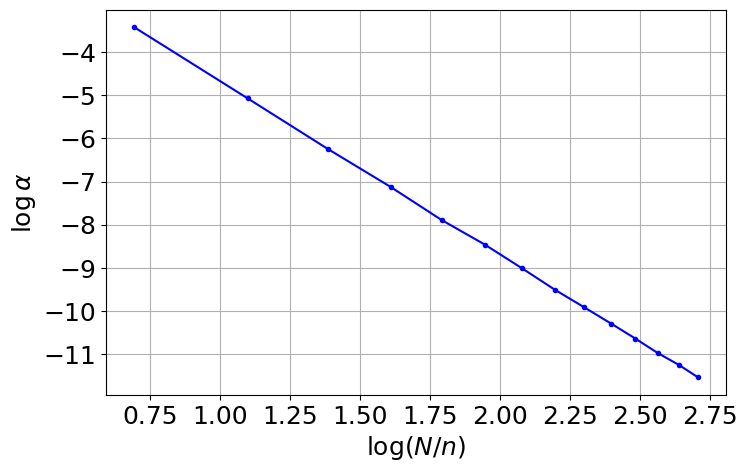}
        \caption{}
        \label{fig:compression2}
    \end{subfigure}
    \caption{(a) Logarithm of the fidelity as a function of the Hamming distance $|s-s'|$, for fixed $n=300$ and $N=800$. (b) $\log\alpha$ as a function of $\log(N/n)$, with $n=300$ and fixed Hamming distance $|s-s'|=40$. The linear fit yields a slope of $\beta\approx-4.02$. }
    \label{fig:compression}
\end{figure}

\bibliographystyle{JHEP}
\bibliography{biblio,charles_ref}

\end{document}